%% file: tracking.tex
\pdfoutput=1
\newif\ifFull
\Fulltrue
\Fullfalse
\documentclass[runningheads]{llncs}

\usepackage[utf8]{inputenc}
\usepackage[english]{babel}
\usepackage{placeins}
\usepackage[noadjust]{cite}  
\usepackage{mdframed}
\usepackage{xcolor}
\usepackage{hyperref}
\usepackage{tabularx}
\usepackage{enumerate}
\usepackage{graphicx}
\usepackage{caption}
\usepackage{subcaption}
\hypersetup{colorlinks = true, citecolor = blue, linkcolor = blue}
\definecolor{blue}{rgb}{0.032812499999999994, 0.3390625, 0.45390624999999996}

\usepackage{amsmath, bm, amssymb}
\usepackage{thmtools}
\usepackage{mathtools}  
\usepackage{thm-restate}  

\spnewtheorem{observation}{Observation}{\bfseries}{\itshape}

\usepackage[capitalize,nameinlink]{cleveref}
\crefname{theorem}{Theorem}{Theorems}
\crefname{lemma}{Lemma}{Lemmas}
\crefname{corollary}{Corollary}{Corollaries}
\crefname{definition}{Definition}{Definitions}
\crefname{claim}{Claim}{Claims}
\crefname{remark}{Remark}{Remarks}
\crefname{observation}{Observation}{Observations}

\newcounter{itemcounter} 
\renewcommand{\emph}[1]{\textbf{\textit{#1}}}
\newcommand{\Paragraph}[1]{\paragraph{\textsf{\textbf{#1}}}}
\newcommand{\todo}[1]{{\color{red}\textbf{TODO(}}#1{\color{red}\textbf{)}}}
\newcommand{\defeq}{\vcentcolon=}
\newcommand{\OPT}{\mathit{OPT}}

\newcommand{\X}{\mathcal{X}}

\newcommand{\C}{\mathcal{C}}
\newcommand{\R}{\mathcal{R}}
\newcommand{\N}{\mathcal{N}}
\newcommand{\U}{\mathcal{U}}

\renewcommand{\problem}[4]{
  \begin{mdframed}%
  \begin{tabularx}{\textwidth}{l X}%
  \multicolumn{2}{l}{#1 #2} \\%
  \textbf{Input:} & #3\\%
  \textbf{Output:} & #4%
  \end{tabularx}%
  \end{mdframed}%
}

\newcommand{\tracking}{\textsc{Tracking}}
\newcommand{\wtracking}{\textsc{WeightedTracking}}
\newcommand{\setcover}{\textsc{SetCover}}

\newcommand{\remove}[1]{}

\title{How to Catch Marathon Cheaters: New Approximation Algorithms 
       for Tracking Paths 
}

\titlerunning{New Approximation Algorithms for Tracking Paths}

\author{
Michael T. Goodrich\inst{1}\orcidID{0000-0002-8943-191X} \and
Siddharth Gupta\inst{2} \and
Hadi Khodabandeh\inst{1} \and
Pedro Matias\inst{1}\orcidID{0000-0003-0664-9145}
}

\authorrunning{M.\,T. Goodrich, S.\,Gupta, H.\,Khodabandeh, and P.\,Matias}

\institute{
Dept. of Computer Science, Univ. of California Irvine, USA
\email{\{goodrich,khodabah,pmatias\}@uci.edu} \and
Dept. of Computer Science, Ben-Gurion Univ. of the Negev, Israel
\email{siddhart@post.bgu.ac.il}
}

\begin{document}

\maketitle

\begin{abstract}\label{abs}
Given an undirected graph, $G$, and vertices, $s$ and $t$ in $G$,
the \emph{tracking paths} problem is that of
finding the smallest subset of vertices
in $G$ whose intersection with any $s$-$t$ path results in
a unique sequence. 
This problem is known to be NP-complete and 
has applications to animal migration tracking and
detecting marathon course-cutting,
but its approximability is largely unknown. 
In this paper, we address this latter issue,
giving novel algorithms having approximation ratios of
$(1+\epsilon)$, $O(\lg \OPT)$ and $O(\lg n)$, for $H$-minor-free, general,
and weighted graphs, respectively. We also give a linear kernel for
$H$-minor-free graphs and make improvements to the quadratic
kernel for general graphs.
\ifFull

\fi
\end{abstract}

\section{Introduction}\label{sec:intro}
\input{intro}

\ifFull\else
\Paragraph{Preliminaries.}
We use standard terminology concerning graphs, 
approximation algorithms and kernelization, which is detailed in \cref{app:terminology}.
For space considerations, content
marked with a link symbol ``\hyperref[app:properties]{$\circledast$}''
is provided in more detail and/or proved in an appendix.
\fi

\ifFull
\section{Preliminaries}\label{sec:prelim}
\input{terminology}
\fi

\section{Structural Properties}\label{sec:properties}
\input{properties}

\section{$H$-Minor-Free Graphs}\label{sec:minor_free}
\input{minor_free}

\section{General Graphs}\label{sec:general}
\input{general_graphs}

\ifFull
\section{Conclusion}

TODO
\fi

\ifFull
\bibliographystyle{splncs04}
\else
\bibliographystyle{abbrvX}
\fi
\bibliography{tracking,minor_free}

\ifFull\else
\clearpage

\begin{appendix}
\input{appendix}
\end{appendix}
\fi

\end{document}

%% file: intro.tex

In most modern marathons, each runner is provided with a small RFID tag,
which is worn on the runner's shoe or embedded in the runner's bib.
RFID readers are placed throughout the course and are 
used to track the progress of the runners~\cite{perv,chok}. 
In spite these measures, some runners try to cheat by taking short
cuts~\cite{wiki:cheat}.
To detect all possible course-cutting,
we are interested in the combinatorial optimization problem of placing
the minimum number of RFID readers in 
the environment of a marathon to determine every possible path from the start to
the finish, including paths that deviate from the official course, just from
the sequence of RFID readers that are crossed by a runner taking a given path.
In addition to detecting marathon course-cutting, 
solutions to this optimization problem could also allow for a
type of marathon where each runner could be allowed to 
map out their own path from
the start to finish so long as their path is at least the required length.

Formally, we model a city road network~\cite{roads,roads2}
through which a marathon will be run as an
undirected graph, $G=(V,E)$, where $V$ is the set of road intersections and
possible RFID reader locations in the city, as well as the placements
of the start and finish lines, and $E$ is the set of road segments joining
two points in $V$ without having any other elements of $V$ in its interior.
Given a start-finish pair, $(s,t)$, of vertices in $G$,
a \emph{tracking set} for $(s,t)$ is a
subset, $T$, of $V$, such that for any $s$-$t$
path\footnote{In this paper, paths do not repeat vertices. We denote a path from $u$ to $v$ by $u$-$v$.} $P$ in $G$,
the sequence $\mathcal{S}^T(P)$ of vertices in $T$ traversed by $P$
uniquely identifies $P$. In other words, $T$ is a tracking set if $\mathcal{S}^T(P)
\neq \mathcal{S}^T(Q)$ for all distinct $s$-$t$ paths $P$ and $Q$.
We formally define the optimization problem,
which is called the \emph{tracking paths} problem, as follows:

\problem{\tracking}{$\!\!(G,s,t)$:}
  {An undirected simple graph $G=(V,E)$ and vertices $s,t\in V$.}
  {A smallest tracking set for $(s,t)$ in $G$.}

We denote by \wtracking{} the vertex-weighted version,
whose goal is to find a tracking set of least total weight. Further, we denote by $k$-\tracking{} the decision version of \tracking{}, which asks whether there exists a tracking set of size at most $k$ (for any given integer $k$). For conciseness,
we refer to 
the ``tracking set of $G$'', when $s$ and $t$ are clear from context.

\ifFull
\input{related_work}
\else
\Paragraph{Related Work.}

\tracking{} has been shown to be NP-Complete
\cite{DBLP:journals/algorithmica/BanikCLRS20}, even when the input
graph is planar \cite{DBLP:conf/isaac/EppsteinGLM19} or has bounded
degree \cite{DBLP:conf/iwoca/Choudhary20}. It is fixed-parameter
tractable (FPT): when parameterized by the solution size (a.k.a., the
natural parameter), it admits a quadratic kernel in general and a
linear kernel when the graph is planar
\cite{DBLP:journals/corr/abs-2001-03161} (other parameterizations
have been studied in \cite{DBLP:conf/ictcs/Choudhary020}). Further,
it admits approximation ratios of $4$ \cite{DBLP:conf/isaac/EppsteinGLM19}
for planar graphs and of $2\Delta+1$ \cite{DBLP:conf/iwoca/Choudhary20}
for degree-$\Delta$ graphs. Exact polynomial time algorithms exist
for bounded clique-width graphs \cite{DBLP:conf/isaac/EppsteinGLM19},
as well as chordal and tournament graphs
\cite{DBLP:conf/iwoca/Choudhary20}. For the NP-hard variant of
tracking only shortest paths between multiple start-finish pairs,
there exists a $O(\sqrt{n\lg n})$-approximation
\cite{DBLP:journals/tcs/BiloGLP20}. 
\fi

\Paragraph{Our Contributions.}
Our results are summarized below:
\begin{enumerate}
  \item \textbf{Linear kernel for $H$-minor-free graphs}. Previously, we only knew of a linear kernel for planar graphs \cite{DBLP:journals/corr/abs-2001-03161}.
  \ifFull This result also immediately implies a $O(1)$-approximation.\fi
  \item \textbf{$(1+\epsilon)$-approximation for $H$-minor-free graphs}. Previous best was a $4$-approximation for planar graphs \cite{DBLP:conf/isaac/EppsteinGLM19}.
  \item \textbf{$O(\lg \OPT)$-approximation for} \tracking{}, where $\OPT$
denotes the cardinality of an optimal tracking set. This is the first algorithm for general graphs with a non-trivial approximation ratio.
\ifFull Previously, we only knew of a $O(\sqrt{n \lg n})$-approximation for tracking shortest paths only
\ifFull and its algorithm requires solving a complex dynamic program \fi
\cite{DBLP:journals/tcs/BiloGLP20}.\fi
  \item \textbf{$O(\lg n)$-approximation for} \wtracking{}. This is the first approximation for weighted graphs, among all variants of \tracking{}\ifFull~above mentioned\fi.
  \item \textbf{Improvements to the quadratic kernel for general graphs of} \cite{DBLP:journals/corr/abs-2001-03161}. We simplify the kernelization algorithm and reduce the constants in the kernel size, while also completing the case analysis in a proof of a lemma central to the kernelization of \cite{DBLP:journals/corr/abs-2001-03161}.
\end{enumerate}

%% file: related_work.tex

\Paragraph{Directly Related Work.}
Motivated from the additional 
applications of tracking animals in migration networks,
tracking intruders in buildings,
and tracking malicious packets in computer networks,
the tracking paths problem was first 
introduced by Banik {\it et al.}
\ifFull
\cite{Banik18,DBLP:journals/algorithmica/BanikCLRS20},
\else
\cite{DBLP:journals/algorithmica/BanikCLRS20},
\fi
who showed that the decision version of this problem is NP-complete,
by a reduction from Vertex Cover and showing
containment in NP by observing that every tracking set must be 
a feedback vertex set in a subgraph which excludes redundant components. 
Recall that
a \emph{feedback vertex set} (FVS) is a set of vertices whose removal results in an acyclic graph. 
In addition, they give the first fixed-parameter tractable (FPT) algorithm,
obtaining a kernel with $O(k^7)$ edges, when parameterized by the solution
size $k$. This has since been improved to $O(k^2)$ edges for general graphs
and $O(k)$ edges for planar graphs~\cite{DBLP:journals/corr/abs-2001-03161}. Other parameterizations also yield FPT algorithms for \tracking{}, including the size of vertex cover and the size of cluster vertex deletion set \cite{DBLP:conf/ictcs/Choudhary020}.
Eppstein {\it et al.}~\cite{DBLP:conf/isaac/EppsteinGLM19} show
that the tracking paths 
problem remains NP-hard even when the input graph is planar
and they give a 4-approximation in
this case. Until now, the only other known approximation algorithm
was a $2(\Delta+1)$-approximation for degree-$\Delta$ graphs
\cite{DBLP:conf/iwoca/Choudhary20}, which is achieved by taking a 2-approximate FVS and all its neighbors.
Optimal polynomial time algorithms have been derived when the graph has bounded clique-width \cite{DBLP:conf/isaac/EppsteinGLM19} (linear time if the corresponding clique-decomposition is given in advance), as well as for chordal and tournament graphs \cite{DBLP:conf/iwoca/Choudhary20}.

\Paragraph{Tracking Shortest Paths.}
The tracking paths problem 
is related to an earlier \emph{tracking set for shortest paths} (TSSP) 
problem, which was
first studied by Banik {\it et al.}~\cite{DBLP:conf/ciac/BanikKPS17}.
In the TSSP problem,
only shortest $s$-$t$ paths need to be tracked, 
thus allowing one to model the input as a directed acyclic graph. 
They show that this variant is APX-hard (and thus, NP-hard),
and they present a 2-approximation for the planar version of 
the TSSP problem,
which is a variant for which we still have no hardness results.
Bil\`{o} et. al. \cite{DBLP:journals/tcs/BiloGLP20} generalized
the TSSP problem, 
allowing for the existence of multiple start-finish pairs and
requiring tracking sets to distinguish between any two shortest
paths between any two start-finish pairs. 
In this setting,
they give a $O(\sqrt{n \log n})$-approximation algorithm. They
further study a version of this problem in which the set of trackers
(ignoring the order in which they are traversed) is itself enough
to distinguish between any two start-finish shortest
paths\footnote{Notice that, when there is a single start-finish
pair, the order of traversed trackers is no longer advantageous,
rendering these two versions of the problem equivalent.}, and they
present a $O(\sqrt{n})$-approximation algorithm in this case.
Moreover, they prove that both of these settings are NP-hard even
for cubic planar graphs, by a reduction from Vertex Cover. Finally,
They also give an FPT algorithm (parameterized by the maximum number of
vertices at the same distance from the start $s$) for the case of a
single start-finish pair, which is the original TSSP
problem introduced by Banik {\it et al.}~\cite{DBLP:conf/ciac/BanikKPS17}.

\Paragraph{Other Related Results.}
Other related work includes work by
Banik and Choudhary \cite{DBLP:conf/caldam/BanikC18}, who consider a
version of \tracking{} on hypergraphs.
which asks for the smallest
subset of vertices whose intersection with each hyperedge is unique.
They prove fixed-parameter tractability of this problem, by showing
a correspondence with the Test Cover problem.

When tracking shortest paths, Banik et. al. \cite{DBLP:conf/ciac/BanikKPS17} also provide a data structure of size $O(n|T|)$ for an $n$-vertex graph and a tracking set $T$ which, given the subset of visited trackers (order does not matter), reconstructs the traversed shortest $s$-$t$ path $P$ in $O(|P|)$ time. They also give an optimal polynomial time algorithm for the related problem \emph{Catching the Intruder}, which asks for the smallest subset of vertices $T$ such that every shortest $s$-$t$ path that visits a vertex of $T$ if and only if it visits a vertex of a given set of forbidden vertices.

When tracking all paths, the authors of \cite{DBLP:conf/iwoca/Choudhary20}
also present a reconstruction algorithm, which, given the set of trackers and
the sequence of traversed trackers, reconstructs the corresponding $s$-$t$ path in polynomial time for a fixed number of trackers. Moreover, they consider a version \tracking{} where one places trackers on weighted edges, instead of (unweighted) vertices. They show that this problem can be solved optimally in polynomial time, by proving that it is equivalent to finding a minimum weighted feedback edge set (i.e. a set of edges whose removal leaves an acyclic graph), which in turn is equivalent to finding a minimum spanning tree.

Also related is the NP-hard problem of finding a minimum \emph{Loop Cutset} (see e.g. \cite{DBLP:journals/ijar/SuermondtC90}), which asks for a subset of vertices in a directed graph of minimum weight which covers every cycle (ignoring edges direction), where a cycle $C$ is covered if we select at least one vertex in $C$ with positive outdegree in $C$. Becker and Geiger \cite{DBLP:conf/uai/BeckerG94} showed that this problem admits a 2-approximation, by reducing it to the minimum weight FVS problem (unfortunately, this reduction does not translate to \tracking{}). Later, Becker \textit{et al.} \cite{DBLP:journals/jair/BeckerBG00} gave a randomized algorithm that outputs the optimal loop cutset with probability at least $1-(1-\frac{1}{6^k})^{c6^k}$ after $O(c\cdot 6^k n)$ steps, where $c>1$ is a constant chosen by the user, $n$ is the number of vertices and $k$ is the size of an optimal loop cutset. Finding small loop cutsets is a crucial step in Pearl's method of conditioning
\ifFull
\cite{DBLP:books/daglib/0066829,DBLP:journals/ai/Pearl86},
\else
\cite{DBLP:books/daglib/0066829},
\fi
an algorithm for probabilistic inference in Bayesian networks (i.e. it computes posterior distribution of variables given new evidence, according to Bayes' Rule).

%% file: terminology.tex

\Paragraph{Notation and Definitions.}

A graph is \emph{simple} if it does not contain any self-loops or parallel edges. We denote by $V(G)$ and $E(G)$ the set vertices and edges, respectively, of a graph $G$. Let us use $G-U$ (resp. $G-u$) to denote the subgraph of $G$ induced by $V(G)\setminus U$ (resp. $V(G)\setminus \{u\})$. The degree of a graph $G$ is the largest degree ${\rm deg}(v)$ among all vertices $v$ in $G$. The neighborhood $N_G(u)$ of a vertex $u$ w.r.t. $G$ is the set vertices of $G$ adjacent to $u$ (when it is clear from the context, the subscript in $N_G(u)$ is omitted). The neighborhood $N_G(U)$ of a vertex set $U$ is simply the union of the neighborhoods for all vertices $u$ in $U$. We denote a bipartite graph by $(U\cup V, E)$, with vertex set partitioned into $U$ and $V$, and edge set $E\subseteq U\times V$. The block-cut tree of a graph is the tree of biconnected components, and a cut-vertex is a vertex shared by at least two biconnected components, such that its removal disconnects the graph (see \cite{DBLP:journals/cacm/HopcroftT73} for more information). Finally, a graph that can be obtained from a graph $G$ by a sequence of edge contractions, edge deletions or vertex deletions is a \emph{minor} of $G$.

\Paragraph{Approximation Algorithms.}

An algorithm is an $\alpha$-\emph{approximation} (algorithm) if it returns a solution $X$ whose cardinality is within $\alpha$ of an optimal solution $X^*$, i.e. $|X|\le \alpha |X^*|$ for minimization problems and $|X|\ge \alpha |X^*|$ for maximization problems. We call $X$ an $\alpha$-\emph{approximate} solution.

\Paragraph{Fixed-parameter Tractability.}

A decision problem parameterized by $k$ admits a \emph{kernel} if there exists a \emph{kernelization} algorithm that outputs, in time polynomial in both $k$ and the size of the instance, a decision-equivalent instance (the kernel) whose size is bounded by $f(k)$, for some computable function $f$. If the kernel size is linear (resp. quadratic) in $k$, we say that the problem admits a \emph{linear} (resp. \emph{quadratic}) kernel. It is well known that a problem admits a kernel if and only if it is fixed-parameter tractable \cite{DBLP:books/sp/CyganFKLMPPS15}, i.e. it can be solved in $g(k)\cdot n^{O(1)}$ time, for some parameter $k$, computable function $g$ and problem size $n$. A kernelization algorithm typically consists of the application of a fixed set of reduction rules, some of which are \emph{rejection} rules -- these return trivial NO-instances, of $O(1)$ size, indicating a negative answer to the decision problem. For more information on kernelization algorithms and parameterized complexity, we refer the reader to
\ifFull
\cite{DBLP:books/sp/CyganFKLMPPS15,DBLP:series/mcs/DowneyF99,DBLP:books/ox/Niedermeier06}.
\else
\cite{DBLP:books/sp/CyganFKLMPPS15}.
\fi


%% file: properties.tex

\begin{definition}[Entry-exit subgraph]\label{def:entry_exit}
Let $(G,s,t)$ be an instance of \tracking{}. An \emph{entry-exit subgraph} is a triple $(G',s',t')$, where $G'$ is a subgraph of $G$, and $(s',t')$ is the \emph{entry-exit pair} corresponding to vertices in $C$ that satisfy the following conditions:
\begin{enumerate}
  \item There exists a path $s$-$s'$ from $s$ to the \emph{entry} vertex $s'$
  \item There exists a path $t'$-$t$ from the \emph{exit} vertex $t'$ to $t$
  \item Paths $s$-$s'$ and $t'$-$t$ are vertex-disjoint
  \item Path $s$-$s'$ (resp. $t'$-$t$) and $G'$ share exactly one vertex: $s'$ (resp. $t'$).
\end{enumerate}
\end{definition}

Notice that the same subgraph $G'$ of $G$ may contain multiple entry-exit pairs\ifFull, hence the triple notation $(G',s',t')$\fi.

\begin{definition}[Entry-exit cycle]
  An \emph{entry-exit cycle} is an entry-exit subgraph $(C,s',t')$, where $C$ is a cycle (see \cref{fig:entry_exit}).
\end{definition}

\ifFull
Intuitively, an entry-exit cycle $(C,s',t')$ (see \cref{fig:entry_exit}) represents a choice between two $s$-$t$ paths and, thus, requires tracking at least one of them.
\fi

We say that a vertex $v$ \emph{tracks} $(C,s',t')$ if $v\in C\setminus \{s',t'\}$. Moreover, we say that $(C,s',t')$ is \emph{tracked} if there exists a tracker in a vertex that tracks it. A cycle $C$ is \emph{tracked} if all entry-exit cycles with entry-exit pairs in $C$ are tracked. If $C$ contains either (i) 3 trackers or (ii) $s$ or $t$ and 1 tracker in a non-entry/non-exit vertex, then it must be tracked. We say that these cycles are \emph{trivially tracked}.

\begin{figure}[b!]
\centering
\includegraphics[scale=.25]{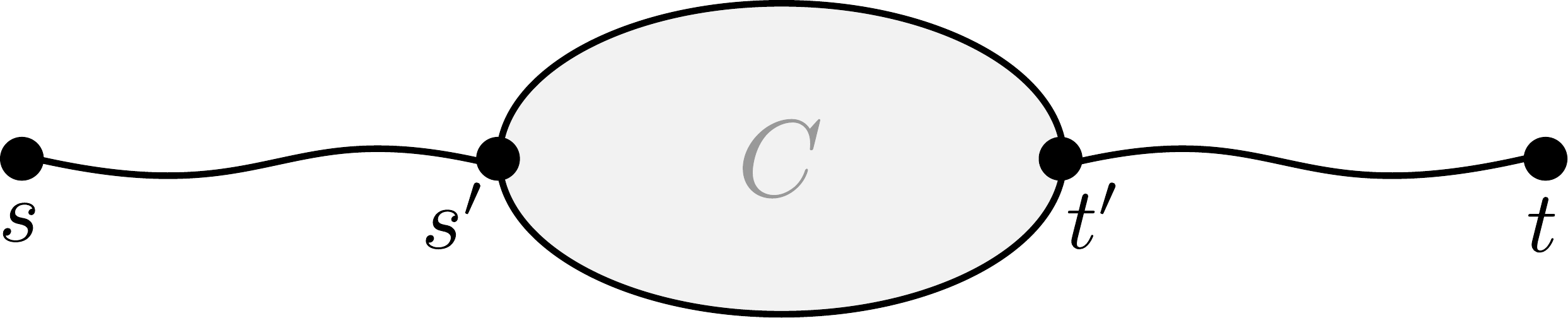}
\caption{Entry-exit pair illustration, with entry vertex $s'$ and exit vertex $t'$.}
\label{fig:entry_exit}
\end{figure}

We rely on the following alternative characterization of a tracking set, due to Banik et. al. \cite[Lemma~2]{DBLP:journals/algorithmica/BanikCLRS20}, which establishes \tracking{} as a covering problem.

\begin{lemma}[\cite{DBLP:journals/algorithmica/BanikCLRS20}] \label{lem:tracking_set}
For a graph $G=(V,E)$, a subset $T \subseteq V$ is a tracking set if and only if every simple cycle $C$ in $G$ is tracked with respect to $T$.
\end{lemma}

\Paragraph{Reduction Rules.}\label{subsec:reduction_rules}

Let us recall some reduction rules previously used to obtain polynomial kernels \cite{DBLP:journals/algorithmica/BanikCLRS20,DBLP:journals/corr/abs-2001-03161} and approximation algorithms \cite{DBLP:conf/isaac/EppsteinGLM19,DBLP:conf/ictcs/Choudhary020,DBLP:conf/ciac/BanikKPS17,DBLP:conf/iwoca/Choudhary20}.

\begin{list}
{\textbf{Rule \arabic{itemcounter}.}}
{\usecounter{itemcounter}\leftmargin=.7in\rightmargin=0em\labelwidth=3in}
  \item \cite{DBLP:journals/algorithmica/BanikCLRS20}
  If there exists an edge or vertex that does not participate in any $s$-$t$ path, remove it from the graph.\label{rule:1}
  \item \cite{DBLP:journals/corr/abs-2001-03161}
  If the degree of $s$ (or $t$) is 1 and $N(s)\neq\{t\}$ ($N(t)\neq \{s\}$), then remove $s$ ($t$), and label the vertex adjacent to it as $s$ ($t$).\label{rule:2}
  \item \cite{DBLP:conf/isaac/EppsteinGLM19}
  If there exist adjacent vertices $a,b\notin \{s,t\}$ such that ${\rm deg}(a)={\rm deg}(b)=2$, then contract the edge $ab$.\label{rule:3}
\end{list}

\begin{definition}
We say that an undirected graph $G$ is \emph{reduced by Rule~X} if it cannot be further by reduced Rule~X. Further, we say that $G$ is \emph{reduced} if it is reduced by \hyperref[rule:1]{Rules~1}, \hyperref[rule:2]{2} and \hyperref[rule:3]{3}.
\end{definition}

After exhaustive application of \hyperref[rule:1]{Rules~1} and \hyperref[rule:2]{2}, the graph is either a single edge, $(s,t)$, or all its vertices have degree at least 2. Henceforth, we assume the latter, since the problem becomes trivial in the former case. \hyperref[rule:3]{Rule~3}, which precludes the existence of adjacent vertices of degree 2, is used to bound the overall number of degree-2 vertices.
\ifFull

\fi
Let us highlight a few additional useful consequences of \hyperref[rule:1]{Rule~1}.

\begin{remark}[{\cite{DBLP:journals/algorithmica/BanikCLRS20}}]\label{rem:subgraph_entry_exit}
Let $G$ be a graph reduced by \hyperref[rule:1]{Rule~1}. Then, every subgraph of $G$ containing at least one edge has at least one entry-exit pair.  
\end{remark}

\begin{remark}[{\cite{DBLP:journals/algorithmica/BanikCLRS20}}]\label{rem:ts_is_fvs}
Let $G$ be a graph reduced by \hyperref[rule:1]{Rule~1}. Then, any tracking
set of $G$ is also an FVS of $G$.
\end{remark}

\begin{remark}\label{rem:block-cut}
  Let $G$ be a graph reduced by \hyperref[rule:1]{Rule~1}. Then the block-cut tree\ifFull\else\footnote{The block-cut tree is the tree of biconnected components.}\fi~of $G$ is an $s$-$t$ path (see \cref{fig:biconnected_cmps}).
\end{remark}

\begin{figure}
\centering
\includegraphics[page=1,scale=1.2]{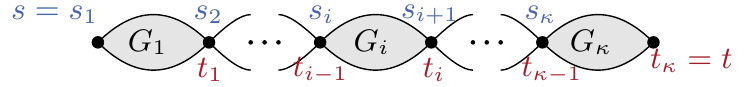}
\caption{The block-cut tree of a graph $G$ reduced by \hyperref[rule:1]{Rule~1} (see \cref{rem:opt_from_biconnected_cmps}). \ifFull An optimal tracking set for $(G,s,t)$ is the disjoint union of all optimal tracking sets for all instances $(G_i,s_i,t_i)$.\fi}
\label{fig:biconnected_cmps}
\end{figure}

In other words, the latter remark says that the graph $G$ that results from exhaustively applying \hyperref[rule:1]{Rule~1} consists of a sequence of $\kappa\ge1$ biconnected components attached together by cut-vertices in a way that is analogous to series composition in series-parallel graphs. Thus, we can turn an instance $(G,s,t)$ of \tracking{} into one or more subproblems on 
biconnected graphs, $(G_i,s_i,t_i)$, one for each biconnected component,
\ifFull
for all $i\in\{1,\dots,\kappa\}$. The cut vertices of $G$ define the start-finish pairs of each subproblem, where $t_i=s_{i+1}$ (for all $i\in\{1,\dots,\kappa-1\}$) and $s_1=s$, $t_\kappa=t$.
\else
as depicted in \cref{fig:biconnected_cmps}.
\fi

\ifFull
Intuitively, a cut-vertex is included in all $s$-$t$ paths, so placing a tracker on it would be helpless, giving us the following:
\fi

\begin{restatable}{remark}{remRuleFour}
\ifFull\else\hyperref[remRuleFourLbl]{$\circledast$}\fi
\label{rem:opt_from_biconnected_cmps}
Let $G$ be a graph reduced by \hyperref[rule:1]{Rule~1}. Then, an optimal
tracking set for $(G,s,t)$ is the \emph{disjoint union} of optimal tracking
sets for all $(G_i,s_i,t_i)$.
\end{restatable}

\ifFull
\begin{proof}
By \cref{rem:block-cut}, an optimal tracking set must contain the union of optimal tracking sets for all $(G_i,s_i,t_i)$. Thus, the remark follows if we show that no optimal tracking set includes a cut-vertex. By \cref{rem:block-cut}, any cut-vertex $v$ of $G$ disconnects the start $s$ from the finish $t$, when removed. It follows that $v$ cannot track any entry-exit cycle, since it will always be entry/exit for any entry-exit cycle containing it.
\qed
\end{proof}
\fi

\Paragraph{Lower Bounds.}\label{subsec:lower_bounds}

We expand on a result by Choudhary and Raman
\cite{DBLP:journals/corr/abs-2001-03161}, which provides a lower bound on the
size of a tracking set, based on the presence of a tree-sink structure in the graph.

\begin{definition}[\cite{DBLP:journals/corr/abs-2001-03161}]
  A \emph{tree-sink} in a graph $G$ is a pair $(Tr,x)$, where $Tr$ is a subtree of $G$ with at least two vertices and $x$, the \emph{sink}, a vertex not in $Tr$ that is adjacent to all the leaves\footnote{We consider a leaf in an unrooted tree to be any vertex of degree 1.} of $Tr$ in $G$. We use $G(Tr,x)$ to denote the subgraph induced by $(Tr, x)$.
  (Notice that this definition does not preclude the adjacency between non-leaf vertices and $x$, as illustrated in \cref{fig:tree-sink}\ifFull\else ~$\hyperref[fig:tree-sink]{\circledast}$\fi.)
\end{definition}

\ifFull
\begin{figure}
\centering
  \includegraphics[page=2,scale=.9]{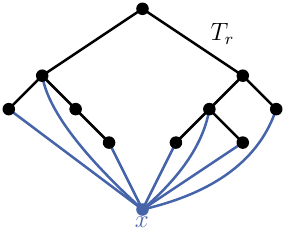}
  \caption{Illustration of a tree-sink $(Tr,x)$.}
  \label{fig:tree-sink}
\end{figure}
\fi

\begin{lemma}[\cite{DBLP:journals/corr/abs-2001-03161}]\label{lem:tree-sink}
  Let $(Tr,x)$ be a tree-sink in a reduced graph $G$, such that
$|N_{Tr}(x)|=\delta$. Further let $(s',t')$ be an entry-exit pair of
$G(Tr,x)$. Then, if $x\in\{s',t'\}$, any tracking set of $G$ contains at least $\delta-1$ vertices in $V(Tr)$.
\end{lemma}

The above lemma is a generalization of the lower bound given by the maximum number of vertex-disjoint paths between any two vertices \cite{DBLP:journals/algorithmica/BanikCLRS20}. We generalize it further to obtain a more useful lower bound, established as the maximum degree among non-cut vertices.

\begin{restatable}{lemma}{lemGraphSink}
\ifFull\else\hyperref[lemGraphSinkLbl]{$\circledast$}\fi
\label{lem:graph-sink}
  Let $G'$ be a subgraph of a reduced graph $G$ and $x$ a vertex in $G'$,
such that $G'-x$ is connected and $N_{G'}(x)=\delta$. Then, any tracking set of $G$ contains at least $\delta-2$ vertices in $G'-x$.
\end{restatable}

\ifFull
\begin{proof}
  XXX
\qed
\end{proof}

\fi

\begin{corollary}\label{cor:max_degree}
  Let $\delta$ be the degree of a non-cut vertex in a reduced graph $G$.
Then, any tracking set of $G$ has size at least $\delta-2$.
\end{corollary}

\Cref{lem:graph-sink} above generalizes \cite[Lemma~8]{DBLP:journals/corr/abs-2001-03161}, and its proof completes the case analysis given in \cite[Lemma~8]{DBLP:journals/corr/abs-2001-03161}. We use it in \cref{sec:minor_free,sec:general}.

%% file: minor_free.tex

A graph is \emph{$H$-minor-free} if it does not contain a fixed graph $H$ as a minor. In this section, we present a linear kernel for $H$-minor-free graphs and use this kernel, as well as some ideas intrinsic to its construction, to design an efficient polynomial-time approximation scheme (EPTAS). An EPTAS is a $(1\pm\epsilon)$-approximate algorithm whose running time is $O(n^c)$ for an input of size $n$ and a constant $c$ independent of $\epsilon$.

\ifFull
This is substantially better than a polynomial-time approximation scheme (PTAS), whose running time can be as bad as $O(n^{f(1/\epsilon)})$, for any computable function $f$ independent of $n$.
\fi
Unlike the minimum FVS problem, which also consists of covering cycles, \tracking{} is not minor-closed \cite{DBLP:journals/corr/abs-2001-03161} (i.e., an optimal solution for a minor of $G$ is not necessarily smaller than an optimal solution for $G$), so the powerful framework of bidimensionality
\ifFull
\cite{DBLP:reference/algo/FominDHT16,DBLP:journals/cj/DemaineH08,DBLP:conf/soda/DemaineFHT04}
\else
\cite{DBLP:reference/algo/FominDHT16}
\fi
cannot be used to obtain either linear kernels
\ifFull
\cite{DBLP:conf/soda/FominLST10,DBLP:reference/algo/Lokshtanov16}
\else
\cite{DBLP:reference/algo/Lokshtanov16}
\fi
or PTASs for $H$-minor-free graphs \cite{DBLP:conf/soda/DemaineH05}. Moreover, \tracking{} does not possess the ``local'' properties required by Baker's technique to develop EPTASs for planar graphs \cite{DBLP:journals/jacm/Baker94}, or apex-minor-free graphs
\ifFull
\cite{DBLP:journals/jgaa/Eppstein99,DBLP:journals/algorithmica/Eppstein00}.
\else
\cite{DBLP:journals/algorithmica/Eppstein00}.
\fi
\ifFull
Intuitively, these local properties can be verified by locally checking constant-sized neighborhoods; nonlocal problems include minimum FVS and minimum connected dominating set, in contrast to problems such as minimum vertex cover, maximum independent set, or minimum dominating set which possess these local properties. Demaine and Hajiaghayi \cite{DBLP:conf/soda/DemaineH05} developed extensions to Baker's approach to obtain PTASs to other nonlocal problems (such as connected dominating set), but these are not applicable to minimum FVS problem, as pointed by the authors.
\fi


\Paragraph{Linear Kernel.}\label{subsec:kernel_minor_free}

\ifFull
We first present a linear kernel for $H$-minor-free graphs.
\fi
The following theorem about the sparsity of $H$-minor-free graphs will be helpful throughout the section.

\begin{theorem}[Mader \cite{mader1968homomorphiesatze}]\label{thm:mader}
Any simple $H$-minor-free graph with $n$ vertices has at most $\sigma_H n$ edges, where $\sigma_H$ depends solely on $|V(H)|$.
\end{theorem}

\ifFull
Since its introduction in~\cite{mader1968homomorphiesatze}, the value of the constant $\sigma_H$ has been improved \cite{kostochka1982minimum,DBLP:journals/combinatorica/Kostochka84,thomason1984extremal}. The current value of $\sigma_H$ is $(\alpha+o(1))|V(H)|\sqrt{\ln |V(H)|}$ \cite{DBLP:journals/jct/Thomason01}, where $\alpha \approx 0.319$.
\fi

We now give the following lemma concerning a relationship between the sizes of the vertex sets in certain bipartite minor-free graphs.

\begin{restatable}{lemma}{lemBipartite}
\ifFull\else\hyperref[lemBipartiteLbl]{$\circledast$}\fi
\label{lem:bipartite}
Let $B=(U\cup V,E)$ be a simple $H$-minor-free bipartite graph, such that:
\ifFull\begin{enumerate}[(i)]\fi
  \ifFull\item \else (i) \fi every vertex in $V$ has degree at least 2, and
  \ifFull\item \else (ii) \fi there exist at most $\delta$ neighbors in common between any pair $u_1,u_2$ in $U$, i.e., $|N(u_1)\cap N(u_2)|\le \delta$ for all $u_1,u_2\in U$.
\ifFull\end{enumerate}\fi
Then $|V|\le\delta\sigma_H|U|$.
\end{restatable}

\ifFull
\begin{proof}
To show the bound on the size of vertex set $V$, we construct a new graph from $B$ as follows. Replace every vertex $v$ in $V$ and its incident edges in $B$ by an edge connecting any two of its neighbors. Observe that, this operation results in an $H$-minor-free graph as it is equivalent to contraction of any edge incident to $v$, followed by the deletion of all but one of the remaining edges incident to $v$). The resulting graph has vertex set $U$, exactly $|V|$ edges, and at most $\delta$ parallel edges between any pair of vertices (this follows from (ii)). By \cref{thm:mader}, any simple $H$-minor-free graph with $|U|$ vertices has at most $\sigma_H|U|$ edges and, thus at most $\delta\sigma_H|U|$ edges when there exist at most $\delta$ parallel edges between any pair of vertices.
\qed
\end{proof}
\fi

Next, we give a lemma which will be useful throughout the paper.

\begin{restatable}{lemma}{lemCutSet}
\ifFull\else\hyperref[lemCutSetLbl]{$\circledast$}\fi
\label{lem:cutSet}
Let $F$ be an FVS of a reduced graph $G$. Then $|V(G - F)| \leq 4|X| - 5$, where $X$ is the cut set defined by $(F, G-F)$, consisting of edges with endpoints in both $F$ and $G-F$.
\end{restatable}

\ifFull
\begin{proof}
Let us partition $V(G-F)$ into $V_1$, $V_2$, $V_{\ge 3}$ corresponding to the sets of vertices whose degree in $G-F$ is (respectively) $1$ (a.k.a. leaves), $2$ or at least $3$. Further, let $V^X$ denote the set of vertices in $V(G-F)$ which are endpoints of an edge in $X$, i.e., the set of vertices in $V(G-F)$ adjacent to a vertex in $F$.

Since each vertex in $V_1$ must be adjacent to a vertex in $F$ ($G$ contains no degree-1 vertices since it's reduced), we have that $V_1\subseteq V^X$ and, thus,
  \[|V_1|\le|X|\]
Moreover, $|V_{\ge 3}|\le |V_1|-2$ (this is a well known fact applicable to any forest). Thus,
\[|V_{\ge 3}|\le |X|-2\]
Next, we bound $|V_2|$. Consider the set $V_2\setminus V^X$ of vertices in $V_2$ which are not adjacent to any vertex in $F$. By \hyperref[rule:3]{Rule~3}, the set $V_2\setminus V^X$ induces an independent set. Hence, $|V_2\setminus V^X|$ is at most the number of edges in the forest that results from replacing each $v \in V_2\setminus V^X$ by an edge connecting $v$'s neighbors, giving us
\[|V_2\setminus V^X| \le |V_1|+|V_{\ge3}|+|V_2\cap V^X|-1,\\\]
which implies
\[|V_2|\le |V_1|+2|V_2\cap V^X|+|V_{\ge3}|-1\]
Thus,
\begin{align*}
  |V(G-F)| & \le 2V_1+2|V_2\cap V^X|+2|V_{\ge3}|-1 \\
            & = 2|V_1\cap V^X|+2|V_2\cap V^X|+2|V_{\ge3}|-1 && (V_1\subseteq V^X)\\
            & \le 2|X|+2|V_{\ge3}|-1 \\
            & \le 4|X|-5
\end{align*}  
\qed
\end{proof}
\fi

We will use \cref{lem:bipartite,lem:cutSet} above to give, in the next lemma, a linear kernel for a biconnected reduced $H$-minor-free graph.

\begin{lemma}\label{lem:linear_size_minor_free_biconnected}
Let $G$ be a biconnected reduced $H$-minor-free graph with start $s$ and finish $t$. Then, $G$ has at most $(16\sigma_H^2 + 8\sigma_H + 1)\OPT-5$ vertices and at most $(20\sigma_H^2+11\sigma_H)\OPT - 6$ edges, where $\OPT$ denotes the size of an optimal tracking set of $G$.
\end{lemma}

\begin{proof}
Let $T^*$ be an optimal tracking set of $(G,s,t)$, i.e., $|T^*|=\OPT$. Note that $G-T^*$ is a forest, since $T^*$ is an FVS of $G$. We assume that $|T^*|\ge 2$, since otherwise one could check, in polynomial time, which vertex of $G$ belongs to $T^*$. We now give some claims about the structure of $G$:
\begin{list}
{\textit{Claim \arabic{itemcounter}:}}
{\usecounter{itemcounter}\leftmargin=1em\labelwidth=4em\labelsep=.5em\itemindent=4em}
\item Let $u_1,u_2$ be two vertices in $T^*$. There exist at most 2 trees in $G-T^*$ that are adjacent\footnote{In this context, a tree is adjacent to $v$ if it includes a vertex that is adjacent to $v$.} to both $u_1$ and $u_2$.
\item Every tree in $G-T^*$ is adjacent to at least 2 vertices in $T^*$.
\item Every tree in $G-T^*$ contains at most 2 vertices adjacent to the same vertex in $T^*$.
\end{list}
The first claim follows from \cref{lem:graph-sink}. If there existed 3 or more trees adjacent to both $u_1$ and $u_2$, then the graph $G'$, induced by $u_1$, $u_2$ and the trees, would require at least 1 tracker in $V(G')\setminus \{u_1\}$ and 1 tracker in $V(G')\setminus \{u_2\}$, contradicting the feasibility of $T^*$. The last claim also follows from \cref{lem:graph-sink} in a similar fashion. The second claim follows from the fact that $G$ is biconnected\ifFull~and hence contains no cut-vertices\fi.

\ifFull
To show the bound on the size of the vertex set and the edge set of $G$, we construct a new graph as follows. We first
\else
Let us
\fi
contract each tree $Tr$ in $G-T^*$ into a \emph{tree vertex} $v_{Tr}$. Let $F$ be the set of all tree vertices. Note that this operation may create parallel edges between a vertex in $T^*$ and a tree vertex, but never between two vertices in $T^*$ or $F$. Furthermore, we remove any edges between vertices in $T^*$. The resulting graph is bipartite, with vertex set partitioned into $T^*$ and $F$, and is $H$-minor-free (since the class of minor-free graphs is minor-closed). By Claims 1 and 2, any 2 vertices in $T^*$ have at most 2 common neighbors, and every vertex in $F$ is adjacent to at least 2 vertices in $T^*$. Hence, by \cref{lem:bipartite}, 
\[|F|\le2\sigma_H|T^*|.\]
As a consequence of Claim 3, there are at most 2 parallel edges between a vertex in $T^*$ and a vertex in $F$. Thus, by \cref{thm:mader}, the set of edges, $X$, in the bipartite graph is at most
\[2\cdot \sigma_H(|F|+|T^*|)\le (4\sigma_H^2+2\sigma_H)|T^*|.\]

Notice that $X$ is the cut set defined by $(T^*,G-T^*)$, consisting of edges with endpoints in both $T^*$ and $G-T^*$. Hence, by \cref{lem:cutSet}, $|V(G-T^*)|\le 4|X|-5$, giving us:
\[|V(G)|\le (16\sigma_H^2 + 8\sigma_H + 1)|T^*|-5.\]
The edges of $G$ consist of (a) edges in $G-T^*$ (at most $|V(G-T^*)|-1$), (b) the cut set $X$, and (c) edges with both endpoints in $T^*$ (at most $\sigma_H|T^*|$ by \cref{thm:mader}). Thus,
\begin{align*}
  |E(G)| & \le (4|X|-6) + |X| + (\sigma_H|T^*|)\\
         & \le (20\sigma_H^2+11\sigma_H)|T^*| - 6.
\end{align*}
\qed
\end{proof}

By \cref{rem:opt_from_biconnected_cmps} and the application of the above lemma to each biconnected component of a reduced graph, we obtain the following.

\ifFull
\begin{lemma}\label{lem:linear_size_minor_free} 
Let $G$ be a reduced $H$-minor-free graph with start $s$ and finish $t$. Then $G$ has at most $(16\sigma_H^2 + 8\sigma_H + 1)\OPT-5$ vertices and at most $(20\sigma_H^2+11\sigma_H)\OPT - 6$ edges, where $\OPT$ denotes the size of an optimal tracking set of $G$.
\end{lemma}

\begin{proof}
The proof follows from \cref{rem:opt_from_biconnected_cmps} and from applying \cref{lem:linear_size_minor_free_biconnected} to each biconnected component of $G$, as in the proof of \cref{lem:quadratic_size_general}.
\qed
\end{proof}
\fi

\ifFull
\begin{theorem}
$k$-\tracking{} admits a linear kernel when restricted to $H$-minor-free graphs.
\end{theorem}
\else
\begin{theorem}\label{thm:linear_size_minor_free} 
  $k$-\tracking{} admits a kernel for $H$-minor-free graphs of size bounded by $(16\sigma_H^2 + 8\sigma_H + 1)k-5$ vertices and $(20\sigma_H^2+11\sigma_H)k - 6$ edges.
\end{theorem}
\fi

\begin{corollary}\label{cor:const_approx}
  \tracking{} admits a $O(1)$-approximation for $H$-minor-free graphs.
\end{corollary}

Even though we develop a $(1+\epsilon)$-approximation in the next section, the latter corollary can be more useful in practice, when running time is a concern.

\Paragraph{EPTAS.}

Given the unsuitability of bidimensionality and Baker's technique discussed earlier, we shall resort to the use of balanced separators.
Our algorithm relies on \ifFull an appropriate decomposition of the graph into regions, which can be accomplished by recursively finding small~\fi\emph{balanced separators}, sets of vertices whose removal partitions the graph into two roughly equal-sized parts.
\ifFull
For simplicity, we define balanced separators in the context of unweighted vertices, but this can be generalized to a weighted setting.

\begin{definition}[Balanced separator]
Let $G$ be a graph, and let $X\subseteq V(G)$ be a subset of its vertices. We say that $X$ is a \emph{balanced separator} if $V(G)\setminus X$ can be partitioned into two sets $A$, $B$ each of which has size at most $2/3|V(G)|$, such that no edge joins a vertex in $A$ with a vertex in $B$.
\end{definition}
\fi
Ungar \cite{ungar1951theorem} first showed that every $n$-vertex planar graph has a balanced separator of size $O(\sqrt{n}\lg^{3/2} n)$.
This was later improved by Lipton and Tarjan \cite{lipton1979separator} to $\sqrt{8n}$, 
and Goodrich~\cite{goodrich1995planar} showed how to compute these 
recursively in linear time.
%
The Lipton-Tarjan separator theorem has been further refined
\ifFull
\cite{chung1988separator,djidjev1982problem,DBLP:journals/jcss/Miller86,DBLP:conf/sigal/GazitM90,DBLP:conf/compgeom/SpielmanT96,DBLP:journals/acta/DjidjevV97,DBLP:journals/jea/Fox-EpsteinMP016}
\else
(e.g., see~\cite{chung1988separator,DBLP:journals/acta/DjidjevV97})
\fi
and generalized to bounded-genus graphs
\ifFull
\cite{DBLP:journals/jal/GilbertHT84,djidjev1985linear,DBLP:journals/siamcomp/Kelner06}
\else
(e.g., see~\cite{DBLP:journals/jal/GilbertHT84,djidjev1985linear})
\fi
as well as to $H$-minor-free graphs
\ifFull
\cite{DBLP:conf/stoc/AlonST90,DBLP:conf/soda/PlotkinRS94,DBLP:journals/talg/ReedW09,DBLP:journals/jacm/BiswalLR10,DBLP:conf/focs/KawarabayashiR10,DBLP:conf/focs/Wulff-Nilsen11}%
\else
(e.g., see~\cite{DBLP:conf/stoc/AlonST90,DBLP:journals/talg/ReedW09})%
\fi.

\begin{theorem}[Minor-free Separator Theorem \cite{DBLP:conf/stoc/AlonST90}]\label{thm:minor_free_separator}
  Let $G$ be an $H$-minor-free graph with $n$ vertices, where $H$ is a simple graph with $h\ge 1$ vertices. Then a balanced separator for $G$ of size at most $c_H^1\sqrt{n}$ can be found in $O(h^{O(1)}n^{O(1)})$ time, where $c_H^1$ is a positive constant depending solely on $h$.
\end{theorem}

We use the Minor-free Separator Theorem recursively to decompose the graph into a set $\R$ of edge-disjoint subgraphs, called \emph{regions}. The vertices of a region $R\in \R$ which belong to at least one other region are called \emph{boundary vertices} and the set of these vertices is denoted by $\partial(R)$. The remaining vertices of $R$ are called \emph{interior vertices} and are denote by $int(R)$.

\begin{definition}[Relaxed $r$-division]
  A \emph{relaxed $r$-division} of an $n$-vertex graph $G$ is a decomposition of $G$ into $\Theta(n/r)$ regions, each of which has at most $r$ vertices, such that the total number boundary vertices is $O(n/\sqrt{r})$\ifFull, for a positive integer $r$\fi.
\end{definition}

\ifFull
A relaxed $r$-division was perhaps among the first applications of Lipton and Tarjan's Planar Separator Theorem, and it was used to obtain EPTASs for maximum independent set \cite{DBLP:journals/siamcomp/LiptonT80} and for minimum vertex cover \cite{chiba1981applications}. A relaxed $r$-division is a relaxed version of an $r$-division, a decomposition introduced by Frederickson \cite{DBLP:journals/siamcomp/Frederickson87} which additionally requires every region to have $O(\sqrt{r})$ boundary vertices. Both decompositions can be constructed in $O(n \lg n)$ time, by recursively splitting the graph using balanced separators. Typically, an $r$-division is constructed from a relaxed $r$-division -- it is shown in \cite{DBLP:journals/siamcomp/Frederickson87} that, after constructing a relaxed division, one can further split the big regions (the ones violating the extra condition) without asymptotically increasing the total number of regions or boundary vertices. Even though $r$-divisions were introduced for planar graphs \cite{DBLP:journals/siamcomp/Frederickson87}, its derivation can easily be generalized to any class of graphs that is characterized by the existence of sublinear balanced separators, including $H$-minor-free graphs.
\else
Computing a relaxed $r$-division is the first step in Frederickson's algorithm \cite{DBLP:journals/siamcomp/Frederickson87} for constructing an $r$-division in a planar graph, a decomposition which additionally requires every region to have $O(\sqrt{r})$ boundary vertices (we won't need this property). Both decompositions can easily be generalized to any class of graphs that is characterized by the existence of sublinear balanced separators, which includes $H$-minor-free graphs.
\fi

\begin{theorem}[Minor-free Separator Theorem (\labelcref{thm:minor_free_separator}) + Frederickson \cite{DBLP:journals/siamcomp/Frederickson87}]
  There is an $O(n\lg n)$ algorithm that, given an $H$-minor-free graph $G$ and a positive integer $r$, computes a relaxed $r$-division of $G$.
\end{theorem}

\ifFull
\begin{proof}
  Follows from Minor-free Separator Theorem (\labelcref{thm:minor_free_separator}) and Frederickson \cite{DBLP:journals/siamcomp/Frederickson87}.
\ifFull\qed\fi
\end{proof}
\fi

\ifFull
A relaxed $r$-division is simpler to construct and will be sufficient for our purposes, so we will use these instead.
\fi

\ifFull
\Paragraph{Algorithm.}
\fi

Our strategy will be to (i) construct a relaxed $r$-division of a
\ifFull
linear kernel $K$ of the input $H$-minor-free graph $G$ (such that $K$ is itself a $O(1)$-approximate tracking set),
\else
smaller graph, $K$, which is itself an $O(1)$-approximate tracking set,
\fi
(ii) solve optimally for each region, and (iii) combine the solutions for each region into a solution for the original graph with quality comparable to that of an optimal solution.
\ifFull
This approach has been used to obtain EPTASs for minimum FVS \cite{DBLP:conf/sea2/BorradaileLZ19,ZhengPhdThesis},
and (without the need for a linear kernelization) maximum independent set \cite{DBLP:journals/siamcomp/LiptonT80}, as well as minimum vertex cover \cite{chiba1981applications}.
\else
This approach has been used to obtain EPTASs for minimum FVS \cite{DBLP:conf/sea2/BorradaileLZ19,ZhengPhdThesis}, maximum independent set \cite{DBLP:journals/siamcomp/LiptonT80} and minimum vertex cover \cite{chiba1981applications}.
\fi
However, and in contrast to these problems, the step of constructing a close to optimal solution from the solutions \ifFull computed\fi of each region is not obvious. Indeed, the difficulty of this step emerges from the very ``nonlocal'' structure of \tracking{}, which
requires special attention to the location of \ifFull the start-finish~\fi $(s,t)$ in the graph%
\ifFull~within the graph,
in addition to requiring tracking cycles that escape out of local neighborhoods (as in minimum FVS). The latter issue complicates not only the intermediate step of solving optimally for each region, but also the step of combining solutions for each region -- an entry-exit cycle spanning two regions may not be tracked by just the boundary vertices it traverses, depending on the position of $s$ and $t$; this is in contrast to minimum FVS, which requires only one tracker per cycle rendering the combining step trivial.
\else
, in addition to the nonlocal structure of cycles, 
as illustrated in \cref{fig:rdivision}\ifFull\else ~$\hyperref[fig:rdivision]{\circledast}$\fi.
\fi
Our EPTAS is as follows:

\begin{mdframed}
\begin{list}
{\textbf{\arabic{itemcounter}.}}
{\usecounter{itemcounter}\leftmargin=1.5em\rightmargin=0em\labelwidth=3in}
  \item Compute a linear kernel $K$ of $G$ by reducing it with \hyperref[rule:1]{Rules~1}, \hyperref[rule:2]{2}, \hyperref[rule:3]{3}\ifFull, such that an optimal tracking set of $K$ is a constant fraction of $K$\fi~(see \cref{cor:const_approx}).
  \item Compute a relaxed $r$-division \ifFull\else $\R$~\fi of $K$ with $r=(2c_H^1c_H^2(c_H^3+1)/\epsilon)^2$, for any choice of $\epsilon>0$ and constants $c_H^1,c_H^2,c_H^3>0$ specified later.
  \ifFull Let $\R$ be the set of resulting regions.\fi
  \item For each region $R$ in $\R$, compute an optimal tracking set $\OPT(R)$ for the subset of entry-cycles (with respect to $(s,t)$) which are completely contained in $R$.
  \item Output $T=\bigcup_{R\in\R}\left(\OPT(R)\cup\partial(R)\cup\N(R)\right)$.

  Here, $\N(R):=N_{\Pi(R)}(\partial(\Pi(R)))$ defines an appropriate neighborhood of the boundary vertices of $R$, where $\Pi(R)$ is the subgraph of $R$ consisting of the union of each path in $R$ that: (i) is not an edge, (ii) has $\partial(R)$ vertices as endpoints, and (iii) traverses no \emph{internal} vertices that are in $\OPT(R)$.
  We let $\partial(\Pi(R))\defeq\partial(R)\cap \Pi(R)$.
  See \cref{fig:Pi_R}.
\end{list}
\end{mdframed}

\begin{figure}[hb!]
  \centering
  \vspace*{-12pt}
  \includegraphics[page=6,scale=1.2]{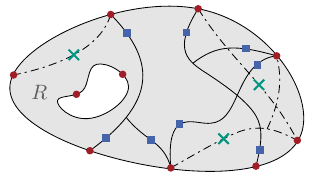}
  \caption{Illustration of $\Pi(R)$ and of $\N(R)$ for a region $R$ in a relaxed $r$-division $\R$. Vertices in $\partial(R)$ are depicted in red circles. $\Pi(R)$ consists of the union of all boundary-to-boundary paths in $R$ (solid black), which are not edges and do not traverse $\OPT(R)$ (green crosses). The dashed lines represent paths in $R-\Pi(R)$. $\N(R)$ is depicted in blue squares.}
  \label{fig:Pi_R}
\end{figure}

\ifFull
Notice that $\bigcup_{R\in \R} \left( \OPT(R)\cup \partial(R) \right)$ may not constitute a tracking set of $K$, because there might exist entry-exit cycles spanning exactly 2 regions, whose entry-exit pairs are the only vertices with trackers, rendering them untracked (see \cref{fig:rdivision}).
\fi

We will now give the details of the algorithm and its correctness. We refer to the \hyperref[rule:1]{Reduction Rules} defined in \cref{subsec:reduction_rules}. As a reminder, after exhaustive application of \hyperref[rule:1]{Rules~1} and \hyperref[rule:2]{2}, the graph is either a single edge between $s$ and $t$, or all its vertices have degree at least 2. Henceforth, we will assume the latter, since a minimum tracking set is trivial in the former.
\ifFull
\hyperref[rule:3]{Rule~3}, which precludes the existence of adjacent vertices of degree 2, is used to bound the overall number of degree-2 vertices.
\fi
\ifFull

\else
\fi
Notice that none of the reduction rules introduce trackers, so there is no lifting required at the end of our algorithm, i.e., adding back any trackers introduced during the reduction.

\begin{observation}\label{obs:tracking_set_K}
No entry-exit cycles are removed during \hyperref[rule:1]{Rules~1}, \hyperref[rule:2]{2} or \hyperref[rule:3]{3}, so a tracking set of the resulting kernel $K$ is a tracking set of the input graph $G$. Therefore, any minimum tracking set of $K$ is also a minimum tracking set of $G$. 
\end{observation}

Next, we explain how to compute in polynomial time optimal tracking sets for each region in a relaxed $r$-division of a kernel $K$.

\begin{restatable}{lemma}{lemOptEachR}
\ifFull\else\hyperref[lemOptEachRLbl]{$\circledast$}\fi
\label{lem:opt_each_R}
  Let $\C(R)$ be the set of all entry-exit cycles in $G$ whose vertices are a subset of $V(R)$, where $R$ is a subgraph of $G$. Then one can compute a minimum subset of $V(R)$ that tracks every entry-cycle of $\C(R)$ in $O(2^{|V(R)|}\cdot n^{O(1)})$ time.
\end{restatable}

\ifFull
\begin{proof}
  It suffices to enumerate all the $2^{|V(R)|}$ possible subsets and, for each, verify in $O(n^{O(1)})$ time whether every entry-exit cycle of $\C(R)$ is tracked. The verification step can be done in a way similar to the verification algorithm given by Banik \textit{et al.} \cite{DBLP:journals/algorithmica/BanikCLRS20} to show that the problem is in NP: from the observation that every tracking set $X$ is an FVS (see \cref{rem:ts_is_fvs}), it follows that there is at most $O(n^{O(1)})$ entry-exit cycles not tracked by $X$ (see also \cref{lem:general:poly_cycles}).
\qed
\end{proof}
\fi

Let us now argue that our algorithm computes a $(1+\epsilon)$-approximate tracking set. Let $T=\bigcup_{R\in \R} \left( \OPT(R)\cup \partial(R) \cup \N(R)\right)$ be the output of the algorithm\ifFull, for the relaxed $r$-division $\R$ of a kernel $K$ of $G$\fi.

\begin{restatable}{lemma}{lemFeasibility}
\ifFull\else\hyperref[lemFeasibilityLbl]{$\circledast$}\fi
\label{lem:feasibility:minor_free}
  $T$ is a tracking set of the input graph $G$.
\end{restatable}

\ifFull
\begin{proof}
It is enough to argue that $T$ tracks every cycle $C$ of $K$ (by \cref{lem:tracking_set} and \cref{obs:tracking_set_K}). We consider 3 types of cycles, illustrated in \cref{fig:rdivision}. If $C$ spans exactly 1 region $R$ in $\R$, then it is guaranteed to be tracked by feasibility of $\OPT(R)$ (see \cref{lem:opt_each_R}). If $C$ spans at least 3 regions in $\R$, then $C$ is trivially tracked by the trackers in, at least, 3 boundary vertices. Otherwise, let $C$ be a cycle spanning exactly 2 regions. We argue that $C$ is also trivially tracked. Let $R_1$ and $R_2$ be the two regions spanned by $C$ and let us assume that $C$ traverses exactly 2 boundary vertices $b_1$ and $b_2$ (if it traverses more boundary vertices, it must be trivially tracked, and if it traverses less, then it could not span more than one region). If $C$ contains a tracker in $\OPT(R_1)\cup\OPT(R_2)\setminus\{b_1,b_2\}$, we are done. Otherwise, $C$ is the union of a path in $\Pi(R_1)$ and a path in $\Pi(R_2)$, by definition of $\Pi$. In this case, however, $C$ must have a third tracker placed in $N(\{b_1,b_2\})\cap \{int(R_1)\cup int(R_2)\}$ (notice that $C$ contains at least one non-boundary vertex that is a neighbor of $b_1$ or $b_2$, since there are no parallel edges).
\qed
\end{proof}
\fi

Let us denote by $\OPT$ the size of an optimal tracking set of the input graph $G$. To argue that $|T|\le (1+\epsilon)\OPT$, we will need to argue that the set of trackers in the special neighborhoods defined by $\N(R)$, for all regions $R$, have small cardinalities, i.e., roughly equal to $O(\epsilon\OPT)$. This is the key argument to our EPTAS, which the next lemma addresses. Its proof is not immediately obvious, since the number of neighbors of all boundary vertices could be $\Omega(\OPT)$, a consequence of the quadratic gap between $|\partial(R)|$ and $|V(R)|$.
\ifFull
Intuitively, if these special neighborhoods were too large, then the following properties could not both be applicable: feasibility of $\OPT(R)$ for every region $R$, and classification of $G$ as minor-free. Our proof below is of combinatorial nature, so it will not quite capture the topological intuition behind it, but it is relatively succinct and avoids complex case analysis.

We will use \cref{lem:tree-sink} to derive a few structural properties of $\Pi(R)$, for each $R\in\R$, and exploit these to bound $|\N(R)|$. As a reminder, $\Pi(R)$ is defined to be the subgraph of $R$ that is the union of each path in $R$ that: is not an edge, has $\partial(R)$ vertices as endpoints, and traverses no \emph{internal} vertices that are in $\OPT(R)$, where $\OPT(R)$ is an optimal tracking set of just the entry-exit cycles of region $R$.
\fi

\begin{restatable}{lemma}{lemSmallNeighborhoods}
\ifFull\else\hyperref[lemSmallNeighborhoodsLbl]{$\circledast$}\fi
\label{lem:small_neighborhoods}
$|\N(R)| \le c_H^3|\partial(\Pi(R))|$, where $c_H^3\ge 9\sigma_H^2+3\sigma_H$.
\end{restatable}

\ifFull
\begin{proof}
The proof is similar in spirit to that of \cref{lem:linear_size_minor_free_biconnected}.
Clearly, $\partial(\Pi(R))$ is an FVS for $\Pi(R)$ (if it were not, there would be untracked cycles in $R$ contradicting feasibility of $\OPT(R)$), so $\Pi(R)-\partial(\Pi(R))$ is a forest. We assume w.l.o.g. that $|\partial(\Pi(R))|\ge 2$. Below, we make some claims about the structure of $\Pi(R)$:
\begin{list}
{\textit{Claim \arabic{itemcounter}:}}
{\usecounter{itemcounter}\leftmargin=1em\labelwidth=4em\labelsep=.5em\itemindent=4em}
\item Let $b_1,b_2$ be two vertices in $\partial(\Pi(R))$. There exist at most 3 trees in $\Pi(R)-\partial(\Pi(R))$ that are adjacent\footnote{In this context, a tree is adjacent to $v$ if it includes a vertex that is adjacent to $v$.} to both $b_1$ and $b_2$.
\item Every tree in $\Pi(R)-\partial(\Pi(R))$ is adjacent to at least 2 vertices in $\partial(\Pi(R))$.
\item Every tree in $\Pi(R)-\partial(\Pi(R))$ contains at most 3 vertices adjacent to the same vertex in $\partial(\Pi(R))$.
\end{list}
The first claim follows from \cref{lem:tree-sink} (if there existed 4 or more trees adjacent to both $b_1$ and $b_2$, there would have to be a tracker from $\OPT(R)$ in one of the trees, a contradiction). The last claim follows from \cref{lem:tree-sink} in a similar fashion. The second claim follows from the definition of $\Pi(R)$ and \hyperref[rule:1]{Rules~1},\hyperref[rule:2]{2}.

Let us contract each tree $Tr$ in $\Pi(R)-\partial(\Pi(R))$ into a \emph{tree vertex} $v_{Tr}$ and let $F$ be the set of all tree vertices. Notice that this may create parallel edges between a vertex in $\partial(\Pi(R))$ and a tree vertex, but never between two vertices in $\partial(\Pi(R))$ or $F$. In addition, let us remove any edges between vertices in $\partial(\Pi(R))$. The resulting graph is bipartite, with vertex set partitioned into $\partial(\Pi(R))$ and $F$, and is $H$-minor-free (since the class of minor-free graphs is minor-closed). By Claims 1 and 2, at most 3 vertices in $F$ share the same pair of neighbors, and every vertex in $F$ has degree at least 2. Hence, by \cref{lem:bipartite}, 
\[|F|\le3\sigma_H|\partial(\Pi(R))|\]
As a consequence of Claim 3, there are at most 3 parallel edges between a vertex in $\partial(\Pi(R))$ and a vertex in $F$. Thus, by \cref{thm:mader}, the set of edges in the bipartite graph is at most
\[3\cdot \sigma_H(|F|+|\partial(\Pi(R))|)\le (9\sigma_H^2+3\sigma_H)|\partial(\Pi(R))|\]

The lemma follows from the fact that the edges in the bipartite graph, including the parallel ones, have a 1-1 correspondence with the vertices in $\N(R)$.
\ifFull\qed\fi
\end{proof}
\else
\begin{proof}
(Sketch)
The set of untracked cycles between 2 regions $R$ and $R'$, which must exist in $\Pi(R) \cup \Pi(R')$, induces a forest on either region if we remove $\partial(R)$ and $\partial(R')$. Using arguments similar to those in the proof of \cref{lem:linear_size_minor_free_biconnected}, we can show that the bipartite graph with bipartition $(F,\partial(\Pi(R)))$ has the properties required by \cref{lem:bipartite}, but also that there exists $O(1)$ edges between a tree and a boundary vertex, where $F$ is the set of trees in $\Pi(R)-\partial(\Pi(R))$. As a consequence, we can get an appropriate bound on the number of edges in this bipartite graph, from which the lemma follows. (See \hyperref[lemSmallNeighborhoodsLbl]{\namecref{app:minor-free}~\labelcref{app:minor-free}} for details.)
\qed
\end{proof}
\fi

Before proving that the output of our algorithm is a $(1+\epsilon)$-approximate tracking set, let us first recall a result from Frederickson \cite[Lemma~1]{DBLP:journals/siamcomp/Frederickson87}
\ifFull\footnote{A more detailed proof of \cite[Lemma~1]{DBLP:journals/siamcomp/Frederickson87}
is given in \cite{planarity.org} for separators based on edge weights.}\fi,
which concerns the sum, for each boundary vertex $b$ of the number of regions $\Delta(b)$ containing $b$ in a relaxed $r$-division $\R$ of a planar graph. Even though this result was given in the context of planar graphs, it can easily be generalized to any graph whose subgraphs $G'$ admit balanced separators of size $O(\sqrt{|V(G')|})$. We denote the set of all boundary vertices by $\partial(\R)$. Further, let $B(\R)=\sum_{b\in \partial(\R)}\left(\Delta(b)-1\right)$.

\begin{lemma}[\cite{DBLP:journals/siamcomp/Frederickson87}]\label{lem:total_boundaries}
  Let $\R$ be a relaxed $r$-division of an $n$-vertex graph whose subgraphs $G'$ admit balanced separators of size at most $c\sqrt{|V(G')|}$. Then $B(\R)\le c\cdot n/\sqrt{r}$, for a constant $c$ independent of $r$ and $n$.
\end{lemma}

We will use the latter lemma to bound the overall number of trackers in the next theorem.

\ifFull
\begin{theorem}
  There exists a $(1+\epsilon)$-approximation algorithm for \tracking{} on $H$-minor-free graphs (for a fixed $H$), running in $O(2^{O(1/\epsilon^2)}n^{O(1)})$ time, for a given $\epsilon>0$.
\end{theorem}
\else
\begin{theorem}
\tracking{} admits an EPTAS for $H$-minor-free graphs.
\end{theorem}
\fi

\begin{proof}
  Consider the algorithm given at the beginning of the section. As a reminder, let $T=\bigcup_{R\in \R} \left( \OPT(R)\cup \partial(R) \cup \N(R)\right)$ be the output of the algorithm, for a relaxed $r$-division $\R$ of a kernel $K$ of $G$, where $\OPT(R)$ is the optimal tracking set computed with respect to entry-exit cycles in $R$. By \cref{lem:feasibility:minor_free}, $T$ is a tracking set. Next, we argue about the approximation ratio.
   By a union bound,
  \[|T|\le |\partial(\R)| + \sum_{R\in\R} |\OPT(R)|+\sum_{R\in\R}|\N(R)|.\]
  Let $n'=|V(K)|$ be the number of vertices in $K$.
  Clearly, $|\partial(\R)|\le B(\R)$. Moreover, we have that $\sum_{R\in\R}|\partial(R)|\le 2B(\R)$, so by \cref{lem:small_neighborhoods}, we have:
    \[\sum_{R\in\R}|\N(R)|\le 2c_H^3B(\R).\]
  Let $T^*$ be an optimal tracking set of $K$, i.e., $|T^*|=\OPT$ (by \cref{obs:tracking_set_K}). Since $T^*$ is a tracking set, but not necessarily an optimal one, for all entry-exit cycles within any region $R\in\R$, we have that $|\OPT(R)|\le |T^*\cap V(R)|$. Thus,
  \[\sum_{R\in\R} |\OPT(R)|\le \OPT + B(\R).\]
  Overall, for $r= (2c_H^1c_H^2(c_H^3+1)/\epsilon)^2$,
  \begin{align*}
    |T| & \le \OPT + 2(c_H^3+1)B(\R)\\
        & \le \OPT + 2c_H^1(c_H^3+1) n'/\sqrt{r} && (\text{\cref{lem:total_boundaries}, \cref{thm:minor_free_separator}})\\
        & \le \OPT + 2c_H^1c_H^2(c_H^3+1)\OPT /\sqrt{r} && 
          (\text{\ifFull\cref{lem:linear_size_minor_free}\else\cref{thm:linear_size_minor_free}\fi}, c_H^2\ge 16\sigma_H^2 + 8\sigma_H + 1)\\
        & = (1+\epsilon)\OPT.
  \end{align*}

  \ifFull
  The first step of the algorithm, concerning the kernelization, takes $O(n^{O(1)})$ time, since it consists of applying \hyperref[rule:1]{Rules~1}, \hyperref[rule:2]{2}, \hyperref[rule:3]{3}. The second step, which computes a relaxed $r$-division can be done in $O(n\lg n)$ time \cite{DBLP:journals/siamcomp/Frederickson87}, using Lipton and Tarjan's linear time algorithm \cite{lipton1979separator} to find a balanced separator. Computing each $\OPT(R)$ in step 3 can be done in $O(2^r\cdot n^{O(1)})$ time, by \cref{lem:opt_each_R}, and thus $O(2^r\cdot n^{O(1)}/r)$ time for all $R\in\R$. Finally, computing $\Pi(R)$ and $\N(R)$ can be done in $O(n^{O(1)})$ time, and thus $O(n^{O(1)}/r)$, for all $R\in\R$. Overall, the total time complexity is is bounded by
  \begin{align*}
    & O(2^r\cdot n^{O(1)})\\
    =\ & O(2^{O(1/\epsilon^2)}\cdot n^{O(1)}).
  \end{align*}
  \else

  Step 1 of the algorithm takes $O(n^{O(1)})$ time, since it consists of applying \hyperref[rule:1]{Rules~1}, \hyperref[rule:2]{2}, \hyperref[rule:3]{3}. Step 2 can be done in $O(n\lg n)$ time \cite{DBLP:journals/siamcomp/Frederickson87}. Step 3 takes $O(2^r\cdot n^{O(1)})$ time, by \cref{lem:opt_each_R}. Finally, step 4 takes $O(n^{O(1)})$ time. Overall, these amount to $O(2^{O(1/\epsilon^2)}n^{O(1)})$.

\qed
\end{proof}

%% file: general_graphs.tex

In this section, we derive an $O(\lg n)$-approximation algorithm for \wtracking{} on general graphs, as well as an $O(\lg \OPT)$-approximation algorithm for \tracking{}.
\ifFull
To the best of our knowledge, the only known approximation ratio for general graphs until now is $O(\sqrt{n\lg n})$ for tracking shortest paths only (as opposed to all paths) with multiple start-finish pairs, a result by Bil\`{o} et. al. \cite{DBLP:journals/tcs/BiloGLP20}.
\fi
In addition, we improve the quadratic kernel of Choudhary and Raman \cite{DBLP:journals/corr/abs-2001-03161} for general graphs and complete the case-analysis of \cite[Lemma~8]{DBLP:journals/corr/abs-2001-03161}.
\ifFull\else

\begin{restatable}{theorem}{thmQuadraticKernel}
\hyperref[thmQuadraticKernelLbl]{$\circledast$}
  $k$-\tracking{} admits a kernel of size bounded by $4k^2+9k-5$ vertices and $5k^2+11k-6$ edges.
\end{restatable}

The latter theorem improves a quadratic kernel of Choudhary and Raman \cite{DBLP:journals/corr/abs-2001-03161}, whose size is bounded by $140k^2-45k$ vertices and $180k^2+65k$ edges.
\fi

\ifFull
\subsection{Approximation Algorithm}\label{subsec:approx_general}
\fi

We reduce an instance $(G,s,t,w')$ of \wtracking{}, for a weight function $w':V(G)\rightarrow \mathbb{Q}$, into an instance $(\U,\X,w)$ of \setcover{}, which asks for the sub-collection of $\X$ of minimum total weight, whose union equals the universe $\U$. Here, $(\U,\X)$ defines a set system, i.e., a collection $\X$ of subsets of a set $\U$, and $w$ is the weight function $w:\X\rightarrow \mathbb{Q}$. It is well known that there exist greedy polynomial-time algorithms achieving approximation ratios of $(1+\ln M)$ \ifFull\cite{DBLP:journals/mor/Chvatal79,DBLP:journals/jcss/Johnson74a,DBLP:journals/dm/Lovasz75}~\fi or of $(1+\Delta)$ \ifFull\cite{DBLP:journals/siamcomp/Hochbaum82}\else\cite{DBLP:books/daglib/0004338,DBLP:books/daglib/0030297}\fi, where $M$ is the size of the largest set in $\X$ and $\Delta$ is the maximum number, over all elements $u$ in $\U$, of sets in $\X$ that contain $u$.
\ifFull
At first glance, this reduction does not seem to yield a useful approximation ratio, given the exponential number of entry-exit cycles that need to be covered, but we will show that one can limit this number to a polynomial of fixed degree.

~
\fi

Let $\C$ be the set of all entry-exit cycles in our input graph $G$, which we assume w.l.o.g. to be reduced by \hyperref[rule:1]{Rule~1}. Further, let $\C_F$ be the set of all entry-exit cycles in $G$, each of which contains at most 2 vertices from the subset $F\subseteq V$. That is, $\C_F\defeq \left \{ (C,s',t')\in \C: |C\cap F| \le 2\right\}$. Our algorithm is as follows.

\begin{mdframed}
\begin{list}
{\textbf{\arabic{itemcounter}.}}
{\usecounter{itemcounter}\leftmargin=1.5em\rightmargin=0em\labelwidth=3in}
  \item Compute a 2-approximate FVS $F$ of $G$ (see \ifFull\cite{DBLP:conf/isaac/BafnaBF95,DBLP:journals/orl/ChudakGHW98,DBLP:conf/uai/BeckerG94}\else\cite{DBLP:books/daglib/0004338,DBLP:books/daglib/0030297}\fi).
  \item Use the greedy algorithm of \ifFull\cite{DBLP:journals/mor/Chvatal79,DBLP:journals/jcss/Johnson74a,DBLP:journals/dm/Lovasz75}\else\cite{DBLP:books/daglib/0004338,DBLP:books/daglib/0030297}\fi~to compute an approximate set covering, $S\subseteq V(G)$, for an instance $(\U,\X,w)$ of \setcover{} where:
  \begin{enumerate}[(i)]
    \item the universe, $\U$, of elements to be covered is $\C_F$
    \item the collection of covering sets, $\X$, is a 1-1 correspondence with $V(G)$, where each covering set with corresponding vertex $v$ is the subset of $\C_F$ which are tracked by $v$, that is,
    \[\X=\{\{(C,s',t') \in \C_F\mid v\text{ tracks } (C,s',t')\}\}_{v\in V(G)}.\]
    \item the weight function $w$ is the weight function $w'$ defined for \wtracking{}, given the 1-1 correspondence between $\X$ and $V(G)$.
  \end{enumerate} 
  \item Output $T=S \cup F$.
\end{list}
\end{mdframed}

\ifFull
copy from appendix
\else
We can show that $|\C_F|=O(n^{O(1)})$. From the observation that every tracking set $F$ is an FVS (see \cref{rem:ts_is_fvs}), it follows that there are at most $O(n^{O(1)})$ entry-exit cycles not tracked by $F$. Thus, our claim follows (details in \cref{app:subsec:approx}).
\fi

\begin{restatable}{theorem}{thmLogNApprox}
\ifFull\else\hyperref[thmLogNApproxLbl]{$\circledast$}\fi
\label{thm:logN-approx}
  ~\wtracking{} admits an $O(\lg n)$-approximation\ifFull algorithm\fi.
\end{restatable}

\Paragraph{Unweighted Graphs.}

We show that the dual of the above set cover formulation has bounded VC-dimension \cite{DBLP:journals/dcg/HausslerW87,vapnik2015uniform}. This immediately improves the approximation ratio to $O(\lg \OPT)$ for \tracking{} (unweighted version) as a consequence of a result by Br{\"{o}}nnimann and Goodrich \cite{DBLP:journals/dcg/BronnimannG95}, which establishes an approximation-ratio of $O(d\lg(dc))$ for unweighted set cover instances with dual VC-dimension $d$ and optimal covers of size at most $c$.

Let $(\U,\X)$ be a set system and $Y$ a subset of $\U$. We say that $Y$ is \emph{shattered} if $\X\cap Y=2^Y$, where $\X\cap Y\defeq \{X\cap Y\mid X\in \X\}$. In other words, $Y$ is shattered if the set of intersections of $Y$ with each $X\in\X$ contains all the possible subsets of $Y$. The set system $(\U,\X)$ has \emph{VC-dimension} $d$ if $d$ is the largest integer for which there exists a subset $Y\subseteq \U$, of cardinality $|Y|=d$, that can be shattered.

The dual problem of an unweighted instance $(\U,\X)$ of \setcover{} is finding a \emph{hitting set} of minimum size, where a hitting set is a subset of $\U$ that has a non-empty intersection with every set in $\X$. In our case, it corresponds to finding the smallest subset of entry-exit cycles that covers every vertex, where a vertex is covered if it tracks least one entry-exit cycle in the subset. This is equivalent to an unweighted instance of \setcover{} with set system $(V, \C_F^*)$, where $V=V(G)$ and $\C_F^*\defeq \{V(C)\setminus\{s',t'\}: (C,s',t')\in \C_F\}$ is the collection of sets, one for each entry-exit cycle, of vertices which can track that entry-exit cycle.

\begin{lemma}
\label{lem:VC_dim}
  The set system $(V,\C_F^*)$ has VC-dimension at most 9.
\end{lemma}

\begin{proof}
  We show that there exists no subset $Y\subseteq V$ of size $|Y|\ge 10$ that can be shattered by $\C_F^*$. Since every element of $\C_F^*$ contains at most 2 vertices from $F$ (by definition of $\C_F$), we cannot have more than 2 vertices from $F$ in $Y$ (since we would then require an entry-exit cycle containing at least 3 vertices in $F$ to shatter $Y$). Thus, the lemma follows if we show that no subset $Y\subseteq V\setminus F$ of size $|Y|\ge 8$ can be shattered by $\C_F^*$. Let us assume, by contradiction, that this is possible. Then, if $Y\subseteq V\setminus F$ is to be shattered by $\C_F^*$, there must exist 2 entry-exit cycles $(C_1,s'_1,t'_1)$ and $(C_2,s'_2,t'_2)$ in $\C_F$ 
(see \cref{fig:vc-dim}\ifFull\else ~$\hyperref[fig:vc-dim]{\circledast}$\fi), such that:
  \begin{itemize}
    \item $C_1$ traverses all vertices of $Y$, say in the order $y_1,y_2,\dots,y_{|Y|}$ (for all $y_j\in Y$), 
    \item $C_2$ traverses every other vertex of $Y$ traversed by $C_1$, say $Y'=\{y_2,y_4,\dots,y_{|Y|}\}$, but not necessarily in the same order (we assume w.l.o.g. $|Y|$ is even).
  \end{itemize}
  Consider the graph consisting of the union of the cycles $C_1,C_2$. Let us contract every shared edge between $C_1,C_2$. Note that $C_1$ remains a cycle that traverses $Y$ and $C_2$ remains a cycle that traverses $Y'$ but not any vertex of $Y\setminus Y'$. So we can safely assume that $C_1$ and $C_2$ do not share any edges. Thus, the union of $C_1,C_2$ is a graph with $|C_1|+|C_2|-|Y|/2$ vertices and $|C_1|+|C_2|$ edges. Since both entry-exit cycles are in $\C_F$, each of $C_1,C_2$ shares at most 2 vertices with $F$. Let us remove such vertices, say there's $k\le4$ of them. The result is a graph with $|C_1|+|C_2|-|Y|/2-k$ vertices and, at best, $|C_1|+|C_2|-2k$ edges (the removed vertices cannot be in $Y$, so they have degree 2). In order for this graph to be acyclic (since $F$ is an FVS by \cref{rem:ts_is_fvs}, and our contractions preserve cycles) we would then require $|Y|<8$ (since any acyclic graph with $n$ vertices has at most $n-1$ edges), a contradiction.
\qed
\end{proof}

The above lemma, combined with the result of Br{\"{o}}nnimann and Goodrich \cite{DBLP:journals/dcg/BronnimannG95} gives us the following.

\begin{theorem}
  \tracking{} admits an $O(\lg \OPT)$-approximation\ifFull algorithm\fi, where $\OPT$ is the size of an optimal tracking set.
\end{theorem}

\ifFull
\subsection{Quadratic Kernel}\label{subsec:kernel_general}
\todo{copy from appendix}
\fi

%% file: appendix.tex

\section{Deferred Figures}

\begin{figure}[hp!]
\centering
\includegraphics[page=2,scale=.9]{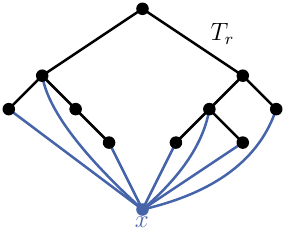}
\caption{Illustration of a tree-sink $(Tr,x)$.}
\label{fig:tree-sink}
\end{figure}

\begin{figure}[hp!]
  \centering
  \includegraphics[page=7,scale=1.4]{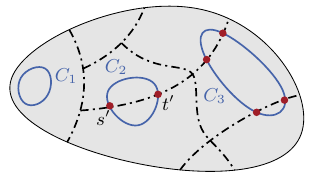}
  \caption{Illustration of a relaxed $r$-division $\R$ (boundaries in dashed lines) and the types of cycles tracked by the output tracking set $T$ (considered in \cref{lem:feasibility:minor_free}). $C_1$-type cycles, which span a single region $R\in\R$ are tracked by $\OPT(R)$. $C_2$-type cycles, which span exactly 2 regions, are not guaranteedly tracked by the 2 boundary vertices they traverse, since these may correspond to an entry-exit pair $(s',t')$. $C_3$-type cycles, which span at least 3 regions, are trivially tracked by the $\ge 3$ boundary vertices they traverse.}
  \label{fig:rdivision}
\end{figure}

\begin{figure}
\centering
  \includegraphics[page=5,scale=1.4]{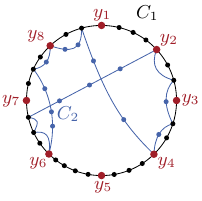}
  \caption{Illustration of proof of \cref{lem:VC_dim}, that the dual VC-dimension is bounded. If a set $Y=\{y_1,y_2,\dots,y_8\}$ is to be shattered by $\mathcal{C}_F^*$, then there must exist cycles $C_1,C_2$ traversing, respectively, $Y$ and every other vertex of $Y$. However, for large enough $Y$ ($|Y|\ge 8$), the existence of $C_1,C_2$ contradicts that $F$ is an FVS.}
  \label{fig:vc-dim}
\end{figure}

\FloatBarrier

\section{Notation and Terminology}\label{app:terminology}

\input{terminology}
\section{Related Work}\label{app:related_work}

\input{related_work}
\section{Deferred Proofs on \hyperref[sec:properties]{Structural Properties}}\label{app:properties}

\remRuleFour*\label{remRuleFourLbl}
\begin{proof}
By \cref{rem:block-cut}, an optimal tracking set must contain the union of optimal tracking sets for all $(G_i,s_i,t_i)$. Thus, the remark follows if we show that no optimal tracking set includes a cut-vertex. By \cref{rem:block-cut}, any cut-vertex $v$ of $G$ disconnects the start $s$ from the finish $t$, when removed. It follows that $v$ cannot track any entry-exit cycle, since it will always be entry/exit for any entry-exit cycle containing it.
\qed
\end{proof}

\lemGraphSink*\label{lemGraphSinkLbl}
\begin{proof}
  Let $Tr$ be a tree in $G'-x$ whose leaves are all adjacent to $x$ (i.e. contained in $N_{G'}(x)$). Such tree can be constructed by trimming a spanning tree of $G'-x$: iteratively remove any leaf that is not adjacent to $x$ in $G'$. Clearly, $(Tr,x)$ is a tree-sink in $G'$, and by \cref{rem:subgraph_entry_exit}, $G(Tr,x)$ has at least one entry-exit pair. If $x$ is in any such entry-exit pair, the lemma follows directly from \cref{lem:tree-sink}, so let us assume otherwise hereafter. 

  Consider the entry-exit pair $(s',t')$ of $G(Tr,x)$ and let us root $Tr$ at $s'$. As in \cite[Lemma~8]{DBLP:journals/corr/abs-2001-03161}, consider the subtrees $Tr_1,Tr_2$ of $Tr$ determined by the edge separator connecting $t'$ to its parent vertex in $Tr$. In particular, let $Tr_1=Tr-Tr(t')$ and $Tr_2=Tr(t')$, where $Tr(v)$ denotes the subtree of $Tr$ rooted at $v\in V(Tr)$. To ensure that every leaf in $Tr_1$ is adjacent to $x$, we again repeatedly remove any leaf of $Tr_1$ that is not adjacent to $x$ (these would correspond to ancestors of $t'$ in $Tr$). Consider the two complementing cases, illustrated in \cref{fig:graph-sink}: (1) both $Tr_1$ and $Tr_2$ have at least one leaf adjacent to $x$ and (2) one of $Tr_1,Tr_2$ has no leaf adjacent to $x$.
  The latter case was neglected in \cite[Lemma~8]{DBLP:journals/corr/abs-2001-03161}.

\begin{list}
{\textbf{Case \arabic{itemcounter}.}}
{\usecounter{itemcounter}\leftmargin=1em\labelwidth=4em\labelsep=.5em\itemindent=4em}
\item Since there is an edge from $x$ to $Tr_2$, there exists an $x-t'$ path which does not intersect $Tr_1$. Thus, $(Tr_1,x)$ constitutes a tree-sink with entry-exit pair $(s',x)$. Similarly, since there is an edge from $x$ to $Tr_2$, $(Tr_2,x)$ constitutes a tree-sink with entry-exit pair $(x,t')$. The lemma follows from applying \cref{lem:tree-sink} to either: (a) each of the tree-sinks $(Tr_1,x)$ and $(Tr_2,x)$, when both $Tr_1$ and $Tr_2$ contain at least two vertices; or (b) to the tree-sink that contains exactly $\delta-1$ leaves adjacent to $x$ (the other tree-sink must be a single vertex when (a) does not hold).
\item Since $Tr$ is rooted at $s'$ and every leaf of $Tr$ is adjacent to $x$,
it must be the case that $Tr_1$ has no leaf adjacent to $x$. Thus, $(Tr_2,x)$
is a tree-sink containing all of $N_{G'}(x)$, so let us apply inductively the
same reasoning we did earlier (with roles of $s',t'$ reversed), whereby we
consider a subdivision of $Tr_2$ into 2 subtrees as established by the
existence of an entry-exit pair $(s'',t'')$ of $G(Tr_2,x)$ -- for simplicity,
we assume that $t''=t'$, since $t'$ is exit for a tree containing $Tr_2$. The
lemma follows directly by the inductive hypothesis that any tracking set of $G$ contains at least $\delta-2$ vertices in $Tr_2$. The base case is a tree-sink corresponding to a star that contains all of $N_{G'}(x)$, and whose leaves are all adjacent to $x$. In this case, an entry-exit pair must belong to the star's root and one of its leaves ($x$ cannot be in an entry-exit pair, given our initial assumption that $x$ did not belong to any entry-exit pair of $Tr$, a supertree of this one). It follows that Case 1 applies to the base case setting.
\end{list}
\qed
\end{proof}

\begin{figure}
  \centering
  \begin{subfigure}[b]{.35\textwidth}
  \centering
    \includegraphics[page=3]{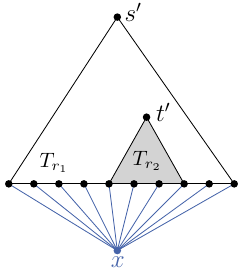}
  \end{subfigure}
  \begin{subfigure}[b]{.35\textwidth}
  \centering
    \includegraphics[page=4]{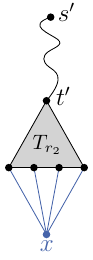}
  \end{subfigure}
\caption{Illustration of proof of \cref{lem:graph-sink}, with Case~1 on the left (both $Tr_1$ and $Tr_2$ have leaves adjacent to $x$) and Case~2 on the right (just $Tr_2$ has leaves adjacent to $x$).}
\label{fig:graph-sink}
\end{figure}

\section{Deferred Proofs on \hyperref[sec:properties]{$H$-Minor-Free Graphs}}\label{app:minor-free}

\subsection{Linear Kernel}

\lemBipartite*\label{lemBipartiteLbl}
\begin{proof}
To show the bound on the size of vertex set $V$, we construct a new graph from $B$ as follows. Replace every vertex $v$ in $V$ and its incident edges in $B$ by an edge connecting any two of its neighbors. Observe that, this operation results in an $H$-minor-free graph as it is equivalent to contraction of any edge incident to $v$, followed by the deletion of all but one of the remaining edges incident to $v$). The resulting graph has vertex set $U$, exactly $|V|$ edges, and at most $\delta$ parallel edges between any pair of vertices (this follows from (ii)). By \cref{thm:mader}, any simple $H$-minor-free graph with $|U|$ vertices has at most $\sigma_H|U|$ edges and, thus at most $\delta\sigma_H|U|$ edges when there exist at most $\delta$ parallel edges between any pair of vertices.
\qed
\end{proof}

\lemCutSet*\label{lemCutSetLbl}
\begin{proof}
Let us partition $V(G-F)$ into $V_1$, $V_2$, $V_{\ge 3}$ corresponding to the sets of vertices whose degree in $G-F$ is (respectively) $1$ (a.k.a. leaves), $2$ or at least $3$. Further, let $V^X$ denote the set of vertices in $V(G-F)$ which are endpoints of an edge in $X$, i.e., the set of vertices in $V(G-F)$ adjacent to a vertex in $F$.

Since each vertex in $V_1$ must be adjacent to a vertex in $F$ ($G$ contains no degree-1 vertices since it's reduced), we have that $V_1\subseteq V^X$ and, thus,
  \[|V_1|\le|X|\]
Moreover, $|V_{\ge 3}|\le |V_1|-2$ (this is a well known fact applicable to any forest). Thus,
\[|V_{\ge 3}|\le |X|-2\]
Next, we bound $|V_2|$. Consider the set $V_2\setminus V^X$ of vertices in $V_2$ which are not adjacent to any vertex in $F$. By \hyperref[rule:3]{Rule~3}, the set $V_2\setminus V^X$ induces an independent set. Hence, $|V_2\setminus V^X|$ is at most the number of edges in the forest that results from replacing each $v \in V_2\setminus V^X$ by an edge connecting $v$'s neighbors, giving us
\[|V_2\setminus V^X| \le |V_1|+|V_{\ge3}|+|V_2\cap V^X|-1,\\\]
which implies
\[|V_2|\le |V_1|+2|V_2\cap V^X|+|V_{\ge3}|-1\]
Thus,
\begin{align*}
  |V(G-F)| & \le 2V_1+2|V_2\cap V^X|+2|V_{\ge3}|-1 \\
            & = 2|V_1\cap V^X|+2|V_2\cap V^X|+2|V_{\ge3}|-1 && (V_1\subseteq V^X)\\
            & \le 2|X|+2|V_{\ge3}|-1 \\
            & \le 4|X|-5
\end{align*}  
\qed
\end{proof}

\subsection{EPTAS}

\lemOptEachR*\label{lemOptEachRLbl}
\begin{proof}
  It suffices to enumerate all the $2^{|V(R)|}$ possible subsets and, for each, verify in $O(n^{O(1)})$ time whether every entry-exit cycle of $\C(R)$ is tracked. The verification step can be done in a way similar to the verification algorithm given by Banik \textit{et al.} \cite{DBLP:journals/algorithmica/BanikCLRS20} to show that the problem is in NP: from the observation that every tracking set $X$ is an FVS (see \cref{rem:ts_is_fvs}), it follows that there is at most $O(n^{O(1)})$ entry-exit cycles not tracked by $X$ (see also \cref{lem:general:poly_cycles}).
\qed
\end{proof}

\lemFeasibility*\label{lemFeasibilityLbl}
\begin{proof}
It is enough to argue that $T$ tracks every cycle $C$ of $K$ (by \cref{lem:tracking_set} and \cref{obs:tracking_set_K}). We consider 3 types of cycles, illustrated in \cref{fig:rdivision}. If $C$ spans exactly 1 region $R$ in $\R$, then it is guaranteed to be tracked by feasibility of $\OPT(R)$ (see \cref{lem:opt_each_R}). If $C$ spans at least 3 regions in $\R$, then $C$ is trivially tracked by the trackers in, at least, 3 boundary vertices. Otherwise, let $C$ be a cycle spanning exactly 2 regions. We argue that $C$ is also trivially tracked. Let $R_1$ and $R_2$ be the two regions spanned by $C$ and let us assume that $C$ traverses exactly 2 boundary vertices $b_1$ and $b_2$ (if it traverses more boundary vertices, it must be trivially tracked, and if it traverses less, then it could not span more than one region). If $C$ contains a tracker in $\OPT(R_1)\cup\OPT(R_2)\setminus\{b_1,b_2\}$, we are done. Otherwise, $C$ is the union of a path in $\Pi(R_1)$ and a path in $\Pi(R_2)$, by definition of $\Pi$. In this case, however, $C$ must have a third tracker placed in $N(\{b_1,b_2\})\cap \{int(R_1)\cup int(R_2)\}$ (notice that $C$ contains at least one non-boundary vertex that is a neighbor of $b_1$ or $b_2$, since there are no parallel edges).
\qed
\end{proof}

\lemSmallNeighborhoods*\label{lemSmallNeighborhoodsLbl}
\begin{proof}
The proof is similar in spirit to that of \cref{lem:linear_size_minor_free_biconnected}.
Clearly, $\partial(\Pi(R))$ is an FVS for $\Pi(R)$ (if it were not, there would be untracked cycles in $R$ contradicting feasibility of $\OPT(R)$), so $\Pi(R)-\partial(\Pi(R))$ is a forest. We assume w.l.o.g. that $|\partial(\Pi(R))|\ge 2$. Below, we make some claims about the structure of $\Pi(R)$:
\begin{list}
{\textit{Claim \arabic{itemcounter}:}}
{\usecounter{itemcounter}\leftmargin=1em\labelwidth=4em\labelsep=.5em\itemindent=4em}
\item Let $b_1,b_2$ be two vertices in $\partial(\Pi(R))$. There exist at most 3 trees in $\Pi(R)-\partial(\Pi(R))$ that are adjacent\footnote{In this context, a tree is adjacent to $v$ if it includes a vertex that is adjacent to $v$.} to both $b_1$ and $b_2$.
\item Every tree in $\Pi(R)-\partial(\Pi(R))$ is adjacent to at least 2 vertices in $\partial(\Pi(R))$.
\item Every tree in $\Pi(R)-\partial(\Pi(R))$ contains at most 3 vertices adjacent to the same vertex in $\partial(\Pi(R))$.
\end{list}
The first claim follows from \cref{lem:tree-sink} (if there existed 4 or more trees adjacent to both $b_1$ and $b_2$, there would have to be a tracker from $\OPT(R)$ in one of the trees, a contradiction). The last claim follows from \cref{lem:tree-sink} in a similar fashion. The second claim follows from the definition of $\Pi(R)$ and \hyperref[rule:1]{Rules~1},\hyperref[rule:2]{2}.

Let us contract each tree $Tr$ in $\Pi(R)-\partial(\Pi(R))$ into a \emph{tree vertex} $v_{Tr}$ and let $F$ be the set of all tree vertices. Notice that this may create parallel edges between a vertex in $\partial(\Pi(R))$ and a tree vertex, but never between two vertices in $\partial(\Pi(R))$ or $F$. In addition, let us remove any edges between vertices in $\partial(\Pi(R))$. The resulting graph is bipartite, with vertex set partitioned into $\partial(\Pi(R))$ and $F$, and is $H$-minor-free (since the class of minor-free graphs is minor-closed). By Claims 1 and 2, at most 3 vertices in $F$ share the same pair of neighbors, and every vertex in $F$ has degree at least 2. Hence, by \cref{lem:bipartite}, 
\[|F|\le3\sigma_H|\partial(\Pi(R))|\]
As a consequence of Claim 3, there are at most 3 parallel edges between a vertex in $\partial(\Pi(R))$ and a vertex in $F$. Thus, by \cref{thm:mader}, the set of edges in the bipartite graph is at most
\[3\cdot \sigma_H(|F|+|\partial(\Pi(R))|)\le (9\sigma_H^2+3\sigma_H)|\partial(\Pi(R))|\]

The lemma follows from the fact that the edges in the bipartite graph, including the parallel ones, have a 1-1 correspondence with the vertices in $\N(R)$.
\qed
\end{proof}

\section{Deferred Proofs on \hyperref[sec:general]{General Graphs}}\label{app:general}

\subsection{Approximation Algorithm}\label{app:subsec:approx}

\thmLogNApprox*\label{thmLogNApproxLbl}

Let us argue about feasibility first. Let $T=S\cup F$ be the output of the algorithm for weighted graphs described in \cref{sec:general}.

\begin{lemma}
\label{lem:feasibility}
  $T$ is a tracking set of $G$.
\end{lemma}

\begin{proof}
  By \cref{lem:tracking_set}, it is enough to argue that every entry-exit cycle of $G$ is tracked by $T$. Let $F$ be the FVS computed in the first step of the algorithm and $S$ the set cover computed in the second step. Since $F\subseteq T$, any entry-exit cycle containing at least 3 vertices from $F$ is trivially tracked. All remaining entry-exit cycles are tracked by $S\subseteq T$, by definition.
\qed
\end{proof}

Next, we argue about the approximation ratio. The lemma below follows from a proof by Banik \textit{et al.} \cite{DBLP:journals/algorithmica/BanikCLRS20} that \tracking{} is in NP.

\begin{lemma}
\label{lem:general:poly_cycles}
  $|\C_F|\le O(n^8)$, where $n$ is the number of vertices of the input graph $G$.
\end{lemma}

\begin{proof}
  Since $F$ is an FVS and $G$ is reduced by \hyperref[rule:1]{Rule~1}, every entry-exit cycle of $\C_F$ includes at least 1 vertex from $F$ (see \cref{rem:ts_is_fvs}) and, by definition, at most 2 vertices from $F$. Since $G-F$ is a forest, there exists at most 1 path between every pair of vertices in $G-F$. Further, each vertex $f$ in $F$ has at most $n-|F|$ neighbors in $G-F$, so there are at most $\binom{n-|F|}{2}$ cycles which contain $f$ and no other vertex from $F$. Thus, the number of cycles containing exactly 1 vertex from $F$ is at most
  \[|F|\binom{n-|F|}{2}\le n^3\]

  Let us now argue about cycles containing exactly 2 vertices $f_1$ and $f_2$ from $F$. Any such cycle is defined by a pair of paths $P$ and $Q$ between $f_1$ and $f_2$ traversing vertices in $V(G-F)\cup \{f_1,f_2\}$. Let us first handle cycles where one of the paths $P,Q$ is a single edge, say $Q$ (because $G$ is simple, the other path, $P$, must consist of at least 2 edges). Clearly, every path $P$ is identified by a path in $G-F$ connecting a neighbor $p_1$ of $f_1$ to a neighbor $p_2$ of $f_2$. Therefore, there exist at most $(n-|F|)^2$ such paths $P$\footnote{In contrast to cycles containing a single vertex from $F$, the neighbors $p_1,p_2$ connected by $P$ may be the same vertex, hence we allow repetitions when counting the number of pairs of neighbors.} and, thus, at most $(n-|F|)^2$ cycles where $f_1$ and $f_2$ are connected by an edge. Similarly, for cycles where $Q$ is not an edge, every path $Q$ is identified by a path in $G-F$ connecting a neighbor $q_1\ne p_1$ of $f_1$ to a neighbor $q_2\ne p_2$ of $f_2$. Hence, there are at most $(n-|F|-1)^2$ such paths $Q$ and, thus, at most $(n-|F|)^2\cdot (n-|F|-1)^2\le (n-|F|)^4$ cycles where $f_1$ and $f_2$ are not connected by an edge. Taking into account all pairs $f_1,f_2$ in $F$, the number of cycles containing exactly 2 vertices from $F$ is at most
  \[\binom{|F|}{2}\left((n-|F|)^2+(n-|F|)^4\right)= O(n^6)\]
   For every cycle $C$, there exist at most $|V(C)|(|V(C)|-1)\le n^2$ entry-exit pairs. Therefore, the number of entry-exit cycles of $C_F$ is at most
   \[n^2(n^3 + O(n^6))=O(n^8)\]
\qed
\end{proof}

Let us denote by $\OPT$ the size of an optimal tracking set of $G$, and let $w(T) = \sum_{a\in T} w(a)$ be the total weight of $T$.

\begin{lemma}
\label{lem:approx}
  $w(T)= O(\lg n)\OPT$.
\end{lemma}

\begin{proof}
  By union bound, we have that $w(T)\le w(S)+w(F)$, where $S$ and $F$ are the sets computed in steps 1 and 2, respectively, of the above algorithm. 

  Since $F$ is a 2-approximate FVS and every optimal tracking set is also an FVS (by \cref{rem:ts_is_fvs}), we have that
  \[w(F)\le 2 \OPT\]

  Let $\OPT_F$ be the size of an optimal solution to the covering problem of step 2, concerning all entry-exit cycles in $\C_F$. Since every optimal tracking set must track all entry-exit cycles in $\C_F$, we have that $\OPT\ge OPT_F$. Further, the well known greedy algorithm for \setcover{} gives us an approximation ratio of at most $(1+\ln |\C_F|)$. Thus, by \cref{lem:general:poly_cycles}, $w(S)=O(\lg n)\OPT_F$ and therefore,
  \[w(S)=O(\lg n)\OPT\]
  The lemma follows.
\qed
\end{proof}

\hyperref[thmLogNApproxLbl]{\namecref{thm:logN-approx}~\labelcref{thm:logN-approx}} follows from \cref{lem:feasibility,lem:approx}.

\subsection{Quadratic Kernel}\label{app:subsec:quadratic}

In this section, we focus on $k$-\tracking{}, the decision version of \tracking{} which asks whether there exists a tracking set of size at most $k$. We consider a parameterization with parameter $k$ itself (often called the \emph{natural} parameter) and give a kernelization algorithm that produces a quadratic kernel for general graphs, by building on the quadratic kernel of Choudhary and Raman \cite{DBLP:journals/corr/abs-2001-03161}. While simpler, our proof of the kernel size completes the case analysis (see \cref{lem:graph-sink}) for one of the lemmas central to the kernelization algorithm of \cite{DBLP:journals/corr/abs-2001-03161} (specifically, \cite[Lemma~8]{DBLP:journals/corr/abs-2001-03161}). Moreover, our kernelization algorithm yields a kernel size with considerably smaller constants. We achieve this by expanding on the notion of tree-sink structures (see \cref{subsec:lower_bounds}), allowing us to bound the maximum degree among non-cut vertices in the kernel (see \cref{cor:max_degree}).



Our kernelization algorithm is simply the exhaustive application of \hyperref[rule:1]{Rules~1}, \hyperref[rule:2]{2} and \hyperref[rule:3]{3} (see \cref{subsec:reduction_rules}) in no particular order, followed by application of \hyperref[rule:4]{Rule~4}, and then of \hyperref[rule:5]{Rule~5}:

\begin{list}
{\textbf{Rule \arabic{itemcounter}.}}
{\usecounter{itemcounter}\leftmargin=.7in\rightmargin=0em\labelwidth=3in}
\setcounter{itemcounter}{3}
  \item
  If there exists a non-cut vertex of degree more than $k+2$, return a trivial NO-instance.\label{rule:4}
  \item
  If the number of vertices (resp. edges) is more than $4k^2+9k-5$ (resp. $5k^2+11k-6$), return a trivial NO-instance.\label{rule:5}
\end{list}

\begin{lemma}\label{lem:rule4}
  \hyperref[rule:4]{Rule~4} is safe and can be done in polynomial-time.
\end{lemma}

\begin{proof}
  Follows from \cref{cor:max_degree}.
\qed
\end{proof}

Next, we show that the last rule is also safe.

\begin{lemma}\label{lem:quadratic_size_biconnected}
  Let $G$ be a biconnected reduced graph, with start $s$ and finish $t$. Then, $G$ has at most $4\OPT^2+9\OPT-5$ vertices and at most $5\OPT^2+11\OPT-6$ edges, where $\OPT$ denotes the size of an optimal tracking set of $G$.
\end{lemma}

\begin{proof}
  Let $T^*$ be an optimal tracking set of $G$, i.e., $|T^*|=\OPT$.
  Since every tracking set is an FVS of a reduced graph (see \cref{rem:ts_is_fvs}), we can apply \cref{lem:cutSet} and obtain $|V(G-T^*)|\le 4|X|-5$, where $X$ is the set of edges with endpoints in both $T^*$ and $G-T^*$. By the fact that $G$ is biconnected and by \cref{cor:max_degree}, $G$ has maximum degree $\OPT+2$ and, hence, $|X|\le \OPT(\OPT+2)$. It follows that
  \[|V(G)|\le 4\OPT^2+9\OPT-5\]
  The edges of $G$ consist of edges with no endpoint in $T^*$ (at most $|V(G-T^*)-1|$) and edges with at least one endpoint in $T^*$ (at most $\OPT(\OPT+2)$ by \cref{cor:max_degree}), giving us
  \[|E(G)|\le 5\OPT^2+11\OPT-6\]
\qed
\end{proof}

We can now apply the latter lemma individually to each biconnected component, giving us the following.

\begin{lemma}\label{lem:quadratic_size_general}
 Any reduced graph $G$ with start $s$ and finish $t$ has at most $4\OPT^2+9\OPT-5$ vertices and at most $5\OPT^2+11\OPT-6$ edges, where $\OPT$ denotes the size of an optimal tracking set of $G$.
\end{lemma}

\begin{proof}
  Let $G_i$ denote the $i\textsuperscript{th}$ biconnected component of $G$, with entry-exit vertices $s_i,t_i$ (see \cref{rem:opt_from_biconnected_cmps}). Further, let $\OPT_i$ denote the size of a minimum tracking set of $(G_i,s_i,t_i)$. It follows from \cref{rem:opt_from_biconnected_cmps} that $\OPT=\sum_{i}\OPT_i$. Moreover, \cref{lem:quadratic_size_biconnected} gives us $|V(G_i)|\le p(\OPT_i)$, where $p(x)=4x^2+9x-5$. Thus,
  \begin{align*}
    |V(G)| & \le \sum_i |V(G_i)|\\
           & \le \sum_i p(\OPT_i) && (\text{\cref{lem:quadratic_size_biconnected}})\\
           & \le p\left(\sum_i \OPT_i\right) && (\text{$p$ is degree-2 polynomial})\\
           & = p(\OPT) && (\text{\cref{rem:opt_from_biconnected_cmps}})
  \end{align*}
  The number of edges in $G$ can be upper bounded in a similar manner.
\qed
\end{proof}

The latter lemma immediately implies the safety of \hyperref[rule:5]{Rule~5}, as well as an $O(\sqrt{n})$-approximation algorithm (output all the vertices in the kernel).

\begin{lemma}\label{lem:rule5}
  \hyperref[rule:5]{Rule~5} is safe and can be done in polynomial-time.
\end{lemma}

Correctness of \hyperref[rule:1]{Rules~1}, \hyperref[rule:2]{2} and \hyperref[rule:3]{3} (\cite{DBLP:journals/algorithmica/BanikCLRS20,DBLP:journals/corr/abs-2001-03161,DBLP:conf/isaac/EppsteinGLM19}), as well as of \hyperref[rule:4]{Rules~4} and \hyperref[rule:5]{5} (\cref{lem:rule4,lem:rule5}) immediately give us the following.

\thmQuadraticKernel*\label{thmQuadraticKernelLbl}

The latter theorem improves a quadratic kernel of Choudhary and Raman \cite{DBLP:journals/corr/abs-2001-03161}, whose size is bounded by $140k^2-45k$ vertices and $180k^2+65k$ edges.